\documentclass[12pt]{article}
\usepackage{amsfonts,amssymb,amsmath,amsthm}
\usepackage{epsfig}
\usepackage{float}

\newcommand{\RR}{\mathbb R}
\newcommand{\CC}{\mathbb C}

\newcommand{\ZZ}{{\mathbb Z}}

\renewcommand{\Re}{\mathop{\mathrm{Re}}}
\renewcommand{\Im}{\mathop{\mathrm{Im}}}

\newtheorem{theorem}{Theorem}
\newtheorem{lemma}{Lemma}
\newtheorem{remark}{Remark}

\newtheorem{example}{Example}

\newcommand{\beq}{\begin{equation}}
\newcommand{\eeq}{\end{equation}}
\newcommand{\ba}{\begin{array}}
\newcommand{\ea}{\end{array}}
\newcommand{\bea}{\begin{eqnarray}}
\newcommand{\eea}{\end{eqnarray}}

\usepackage{graphicx}
\usepackage[usenames,dvipsnames]{color}
\usepackage{transparent} 

\DeclareMathAlphabet{\mathpzc}{OT1}{pzc}{m}{it}

\usepackage{ulem}
\usepackage{cancel}

\begin{document}

\begin{center}
  The finite gap method and the periodic Cauchy problem\\ of $2+1$ dimensional anomalous waves\\ for the focusing Davey-Stewartson 2 equation. 1
\vskip 10pt
{\it P. G. Grinevich $^{1,3}$ and P. M. Santini $^{2,4}$}

\vskip 10pt

\vskip 10pt

{\it 
$^1$ Steklov Mathematical Institute of Russian Academy of Sciences, 8 Gubkina St., Moscow, 199911, Russia.

\smallskip

$^2$ Dipartimento di Fisica, Universit\`a di Roma "La Sapienza", and \\
Istituto Nazionale di Fisica Nucleare (INFN), Sezione di Roma, \\ 
Piazz.le Aldo Moro 2, I-00185 Roma, Italy}

\vskip 10pt

$^{3}$e-mail:  {\tt pgg@landau.ac.ru}\\
$^{4}$e-mail:  {\tt paolomaria.santini@uniroma1.it \\ paolo.santini@roma1.infn.it}
\vskip 10pt

{\today}

\end{center}

\begin{abstract}
The focusing Nonlinear Schr\"odinger (NLS) equation is the simplest universal model describing the modulation instability (MI) of $1+1$ dimensional quasi monochromatic waves in weakly nonlinear media, and MI is considered the main physical mechanism for the appearence of  anomalous (rogue) waves (AWs) in nature. In analogy with the recently developed analytic theory of periodic AWs of the focusing NLS equation, in this paper we extend these results to a $2+1$ dimensional context, concentrating on the focusing Davey-Stewartson 2 (DS2) equation, an integrable $2+1$ dimensional generalization of the focusing NLS equation. More precisely, we use the finite gap theory to solve, to leading order, the doubly periodic Cauchy problem of the focusing DS2 equation for small initial perturbations of the unstable background solution, what we call the periodic Cauchy problem of the AWs. As in the NLS case, we show that, to leading order, the solution of this Cauchy problem is expressed in terms of elementary functions of the initial data.
\end{abstract}

\section{Introduction}

Anomalous waves (AWs), also called rogue, freak waves, are extreme waves of anomalously large amplitude with respect to the surrounding waves, arising apparently from nowhere and disappearing without leaving any trace. Deep sea water is the first environment where AWs have been studied, and the term rogue wave (RW) was first coined by oceanographers (a RW can exceed the height of 20 meters and can be very dangerous). Although the existence of AWs was notorious even to the ancients, and it is a recurring theme in the history of literature, the first scientific observation and measurement of an AW has been made only in 1995 at the Draupner oil platform in the North Sea \cite{Haver}. Starting from the pioneering work on optical fibers \cite{Solli}, it was understood that AWs are not confined to oceanography, and their presence is ubiquitous in nature: they have been observed or predicted in nonlinear optics \cite{Solli,Kibler,Kibler2,Akhmed2013,Liu}, in Bose-Einstein condensates \cite{Wen,Bludov}, in plasma physics \cite{Bailung,Moslem}, and in other physical contexts. In the presence of nonlinearity, the accepted explanation for AW formation is the modulational instability (MI) of some basic solutions, first discovered in nonlinear optics \cite{Bespalov} and ocean waves \cite{BF,Zakharov}, but associated with the appearance of AWs only in the last 20 years. In the understanding of the analytic properties of AWs, it is of great importance the fact that the nonlinear stages of MI in 1+1 dimensions are well described by exact solutions of the integrable \cite{ZSh1,ZSh2} focusing Nonlinear Schr\"odinger (NLS) equation in 1+1 dimensions
\begin{equation}\label{eq:DS_q1}
i u_t +u_{xx}+2 |u|^2 u=0, \ \ u=u(x,t)\in\CC, \ \ x,t\in\RR
\end{equation}
known as breathers \cite{Kuznetsov, Peregrine, Akhmed0}, which can be considered as prototypes of AWs \cite{Dysthe,Henderson,Osborne,ZakharovGelash2,Akhmed2013,Baronio}. Such solutions can be reproduced in wave tank laboratories, optical fibers and photorefractive crystals with a high degree of accuracy  \cite{Chabchoub,Kibler,PieranFZMAGSCDR}.

Using the finite-gap method, the NLS Cauchy problem for periodic initial perturbations of the unstable background, what we call the Cauchy problem of AWs, was recently solved to leading order \cite{GS1,GS5} in the case of a finite number of unstable modes, leading to a quantitative description of the recurrence properties of the multi-breather generalization \cite{Its} of the Akhmediev breather \cite{Akhmed0}. In the simplest case of one unstable mode only, this theory describes quantitatively a Fermi-Pasta-Ulam-Tsingou recurrence of AWs described by the Akhmediev breather (AB) \cite{GS1,GS2}. In addition, a finite-gap perturbation theory for 1+1 dimensional AWs has been also developed \cite{Coppini1,Coppini2} to describe analytically the order one effect of small physical perturbations of the NLS model on the AW dynamics. See also \cite{Coppini3}.

We remark that in this paper the term AW is used in an extended sense, and by it we just mean order one (or higher order) coherent structures over the unstable background, generated by MI, and the above results on NLS AWs deal with the analytic aspects of the deterministic theory of periodic AWs, for a finite number of unstable modes. A coherent structure of anomalously large amplitude with respect to the average one arises from the event in which the nonlinear interaction of many unstable modes is constructive. Therefore the formation of an AW is, strictly speaking, a statistical event. Statistical aspects of the theory of NLS AWs in 1+1 dimensions can be found in  \cite{Gelash-stat, Dematteis-stat, El1-stat,El2-stat}. 

While optical fibers offer a natural testbed for 1+1 dimensional AWs, ocean wind waves are intrinsically two dimensional and it is not yet clear to what extent the above analytic solutions of the 1+1 NLS equation can be observed in the ocean and are relevant in multidimensional nonlinear optics \cite{Fedele,Onorato1}. A good understanding of the deterministic and statistical properties of AWs in 2+1 dimensions is to date lacking; the main difficulty in generalizing the 1+1 theory to a multidimensional context is related to the fact that the large majority of physically relevant 2+1 generalizations of the NLS equation are not integrable, and exact solution techniques are not available. We have in mind, for instance, the elliptic NLS equation, relevant in nonlinear optics \cite{Newell,Kivshar}, the hyperbolic NLS equation, relevant in deep-water gravitational waves \cite{Zakharov},  and the majority of Davey-Stewartson type equations \cite{DS}, relevant in nonlinear optics, water waves, plasma physics and Bose condensates \cite{Benney,DS,ABB,Nishinari,Huang}.

The integrable Davey-Stewartson (DS) equations \cite{Anker,AH} can be written in the following form
\begin{equation}\label{eq:DS1}
  \begin{array}{l}
    i u_t +u_{xx}-\nu^2 u_{yy}+2\eta q u=0, \ \  \eta = \pm 1, \ \ \nu^2 = \pm 1, \\
    q_{xx}+\nu^2 q_{yy}= \left(|u|^2\right)_{xx} - \nu^2\left(|u|^2\right)_{yy},  \\
  u=u(x,y,t)\in\CC, \ \ q(x,y,t)\in\RR, \ \ x,y,t\in\RR,  
  \end{array}
\end{equation}
where $u$ is the complex amplitude of the monochromatic wave, and the real field $q(x,y,t)$ is related to the mean flow. If $\nu^2=-1$ we have the DS1 equation (surface tension prevails on gravity in the water wave derivation); in this case the sign of $\eta$ is irrelevant, since one can go from the equation with $\eta=-1$ to the equation with $\eta=1$ via the changes $q\to -q$ and $x\leftrightarrow y$; therefore there exists only one DS1 equation. If $\nu^2=1$, gravity prevails on surface tension and we have the DS2 equations; in this case the sign of $\eta$ cannot be rescaled away and we distinguish between focusing and defocusing DS2 equations for respectively $\eta=1$ and $\eta=-1$. It turns out that the shallow water limit of the Benney-Roskes equations \cite{Benney} leads to the DS1 and to the defocusing DS2 equations \cite{AS}. 

Although some examples of AW exact solutions of the DS equations \eqref{eq:DS1} are present in the literature (see for example \cite{Liu1,Liu2,Otha1, Otha2}), their importance in physically relevant Cauchy problems is still to be understood. 

As we shall see in the next section, i) the doubly periodic Cauchy problem is well posed for the focusing and defocusing DS2 equations; ii) a linear stability analysis implies that, as in the NLS case, the DS2 background solution is linearly stable in the defocusing case $\eta=-1$, and unstable for sufficiently small wave vectors in the focusing case $\eta=1$. It follows that, although its physical relevance is not clear at the moment, the integrable focusing DS2 equation (with $\nu^2=\eta=1$) is the best mathematical model on which to construct an analytic theory of $2+1$ dimensional AWs. This is the goal of this paper.

  The focusing DS2 equation has also important applications to the differential geometry of surfaces in $\RR^4$. The fact that the surfaces in $\RR^3$ can be locally immersed using squared eigenfunctions of the 2-d Dirac operator with real potential $u(x,y)$, and that the modified Novikov-Veselov hierarchy {\color{Violet} (which is a part of the focusing DS2 hieararchy with extra reality reduction)} acts on such immersions was pointed out in \cite{Konopel1}. Extension of this construction to the surfaces in $\RR^4$ was suggested in \cite{PedPink}, see also \cite{Konopel2}. For the surfaces in $\RR^4$ the symmetries of immersions are generated by the DS2 hierarchy.

  If the surface is compact and has the topology of a torus, the corresponding Dirac operators are doubly-periodic. The existence of a global representation for immersions of tori into $\RR^3$ was proved in \cite{Taim1}, see also \cite{Taim2,Taim3}. In this case, in particular, the famous Willmore functional coincides with the energy conservation law for the  Novikov-Veselov hierarchy \cite{Taim1}; therefore is is possible to apply the soliton theory to classical geometrical problems, see the review  \cite{Taim6} and references therein for more details. Let us point out that, in contrast with the 1-dimensional case, the analytical aspects of the spectral theory in 2-dimensions are deeply non-trivial; for example, the proof of the existence of the zero-energy spectral curve, obtained in \cite{Taim2}, see also \cite{Taim6}, is based on a deep result  of Keldysh. Let us remark that a different approach for constructing the spectral curve, providing more information about its asymptotic properties, was developed in \cite{Krichever2,Krichever3}.    

{\color{Violet}
Let us also remark that, in contrast with the immersions into $\RR^3$, the parameterization of the surfaces in $\RR^4$ in terms of the Dirac operator is not unique, and different representations may generate different dynamics of surfaces. In addition, a priory it is not clear if the DS dynamics is compatible with the periodicity. These questions were discussed in \cite{Taim5}, where, in particular, it was shown that it is possible to define dynamics preserving the conformal classes of the toric surfaces as well as their Willmore functional.
}  
  
In 2021 professor Iskander Taimanov celebrated his 60th birthday. To celebrate his very important results in the area of mathematical physics, and, in particular, his works dedicated to the doubly-periodic theory for the 2-d Dirac operator and DS2 equation, and their applications to the geometry,  we would like to dedicate this article to his 60th anniversary.

\section{Preliminaries}

The DS equations \eqref{eq:DS1} can be written as compatibility condition for the following pair of auxiliary linear operators:
\begin{align}
\label{eq:DS_lax}
  \nu\vec\psi_y &= i\sigma_3\vec\psi_x + U\vec\psi,\nonumber\\
  \vec\psi_t &= 2i\sigma_3\vec\psi_{xx} + 2U\vec\psi_x + V\vec\psi,
\end{align}
where
\begin{align}
\label{eq:DS_lax2}
U &= \begin{pmatrix} 0 & u \\ -\eta \bar u & 0  \end{pmatrix}, \ \
                                            V = \begin{pmatrix} -\eta(w-iq) & u_x-i\nu u_y \\ -\eta (\bar u_x+i\nu \bar u_y) &  -\eta(w+iq)  \end{pmatrix},\nonumber\\
  \nu w_y &= (q -|u|^2 )_x, \ \ w_x=-\nu (q+|u|^2)_y.
\end{align}

Equations (\ref{eq:DS1}) are non-local, and the function $q$ is defined up to an arbitrary integration constant
\begin{equation}
 \label{eq:DS_q1} 
q (x,y,t) \rightarrow  q(x,y,t) + f(t).
\end{equation}
Let us point out that this change of integration constant corresponds to the standard gauge transformation:

\begin{equation}
 \label{eq:DS_gauge} 
u(x,y,t) \rightarrow u(x,y,t) \exp\left(-i\frac{\eta}{2}\int^t f(\tau) d\tau   \right).
\end{equation}

The DS equation has the following real forms:
\begin{enumerate}
\item  Davey-Stewardson I (DS1): $\nu=i$, i.e. $\nu^2=-1$;
\item  Davey-Stewardson II (DS2): $\nu=1$, i.e. $\nu^2=1$.  
\end{enumerate}
By analogy with Nonlinear Schr\"odinger equation, the DS2 equation can be either \textbf{self-focusing} ($\eta=1$) or
 \textbf{defocusing} ($\eta=-1$) .

In the Fourier representation we have:
\begin{equation}
 \label{eq:DS_q2} 
\hat q (\vec k) = Q\cdot \widehat{|u|^2}(\vec k), \ \  Q=\frac{k_x^2- \nu^2k_y^2 }{k_x^2+\nu^2k_y^2 }.
\end{equation}

In the DS1 case the linear operator $Q$ is unbounded, and it results in very non-trivial analytic effects. For the infinite $(x,y)$-plane problem  one has to impose additional boundary conditions at infinity; in particular, by a proper choice of the boundary conditions at infinity, one can generate exponentially localized solutions (dromions) \cite{Fokas1,Fokas2}, first discovered via B\"acklund transformations \cite{BLMP}. In the doubly-periodic case it is not clear if the DS1 system is well defined, and as far as we know, this problem was not properly studied in the literature.

In contrast with the DS1 case, for the DS2 with spatially doubly-periodic boundary conditions, the linear operator $Q$ is well-defined and has unit norm on the space of functions with zero mean value $L^2_0(T^2)$. It can be naturally extended to a linear map
\begin{equation}
\label{eq:DS_q3} 
Q:L^2(T^2)\rightarrow L^2_0(T^2),
\end{equation}
by assuming that
\begin{equation}
\label{eq:DS_q4} 
\iint_{T^2} q(x,y,t) dx dy = 0, \ \ \mbox{for all} \ \ t.
\end{equation}
Due to the gauge freedom (\ref{eq:DS_gauge}), constraint  (\ref{eq:DS_q4}) is not restrictive, and after imposing it, the DS2 flow becomes well-defined.  

Let us recall that 2+1 integrable systems are usually non-local. The most famous example is the Kadomtsev-Petviashvily hierarchy. 

Consider a small perturbation of a constant DS2 solution:
\begin{equation}
\label{eq:DS_q5} 
u(x,y,0) = a + \varepsilon v(x,y).
\end{equation}
The linearized DS2 equation has the following form:
\begin{equation}
\label{eq:DS_lin1} 
i v_t + v_{xx}-v_{yy} + 2 \eta \, Q\cdot [|a|^2 v + a^2 \bar v].
\end{equation}
For a monochromatic perturbation
\begin{equation}
\label{eq:DS_q7} 
v(x,y,t)=  u_1 \exp(i[k_x x + k_y  y]+ \sigma t) + u_{-1} \exp(-i[k_x x + k_y  y]+ \sigma t),
\end{equation}
we obtain
\begin{align}
\label{eq:DS_lin2} 
&[i\sigma -k_x^2+k_y^2] u_1 + 2 \eta \frac{k_x^2- k_y^2 }{k_x^2+ k_y^2 }[|a|^2 u_1 + a^2 \bar u_{-1} ], \\ 
&[i\sigma +k_x^2-k_y^2] \bar u_{-1} - 2 \eta \frac{k_x^2- k_y^2 }{k_x^2+ k_y^2 }[|a|^2 \bar u_{-1} + \bar a^2 u_{1} ].
\end{align}
therefore
\begin{align}
\label{eq:DS_lin3} 
\sigma = \pm \frac{(k_x^2- k_y^2 )\sqrt{4\eta |a|^2 - (k_x^2+ k_y^2)}} {\sqrt{k_x^2+ k_y^2 }}.
\end{align}
We see that in the defocusing case $\eta=-1$ the increment $\sigma$ is imaginary for all wave vectors $(k_x,k_y)$ and the constant solution is linearly stable. In the self-focusing case  $\eta=1$ the harmonic perturbations inside the disk $k^2 \le 4 |a|^2$ are unstable and outside this disk harmonic perturbations are stable. In the doubly-periodic problem we have a finite number of unstable modes.

Therefore, from the point of view of the anomalous waves theory, the most important real form of the DS equation is the focusing DS2 equation:  $\nu=1$, $\eta=1$. This real form also appears in the theory of surfaces in $\RR^4$ defined using the Generalized Weierstrass representation, see \cite{Taim6} and references therein.

{\color{Violet}
In contrast with the focusing NLS equation, DS2 solutions with smooth Cauchy data may blow up in finite time {\color{red}\cite{Ozawa,Otha2}}, and this fact may be important in physical applications. This problem is discussed, in particular, in the book \cite{KleinSaut}. It is possible to construct blow-up solutions using the Moutard transformations; for the modified Novikov-Veselov equation it was done in \cite{TaimTz}. A beautiful geometric model for this generation of singularities was suggested in \cite{MatTaim,Taim9,Taim10}. Consider an immersion of a surface in $\RR^4$. It is possible to construct explicitly the Moutard transformation corresponding to the inversion of the space $\RR^4$. If the original family of surfaces passes through the origin, after inversion the family blows up for some $t=t_0$, as well as the corresponding DS2 solution.

In our text we do not discuss the blow-up of the DS2 solutions corresponding to the Cauchy problem of anomalous waves. We postpone this interesting issue to a subsequent paper.
}

\section{AWs doubly-periodic Cauchy problem for the focusing DS2 equation} 

We study the spatially doubly-periodic Cauchy problem for the focusing DS2 equation
\begin{align}
\label{eq:DS2}
&iu_t + u_{xx} - u_{yy} + 2 q u = 0, \nonumber\\
&q_{xx} + q_{yy} = \left(|u|^2\right)_{xx} - \left(|u|^2\right)_{yy},\nonumber\\
  &u = u(x, y, t) \in \CC,  \ \ q = q(x, y, t) \in \RR,\\
  &u(x+L_x,y,t)=u(x,y+L_y,t)=u(x,y,t), \nonumber\\
  &q(x+L_x,y,t)=q(x,y+L_y,t)=q(x,y,t),\nonumber
\end{align}
assuming that the Cauchy data is a small perturbation of a constant solution
\begin{align}
\label{eq:DS2_1}
& u(x,y,0) = a + \varepsilon v_0(x,y), \ \ \varepsilon\in\RR, \ \ \varepsilon \ll 1,\\
  &v_0(x+L_x,y)=v_0(x,y+L_y)=v_0(x,y). \nonumber
\end{align}
We call equation (\ref{eq:DS2}) with the initial data  (\ref{eq:DS2_1}) the \textbf{doubly-periodic Cauchy problem for anomalous waves.} 
The first auxiliary linear problem can be rewritten as:

\begin{equation}
\label{eq:DS_lax3}
  \begin{bmatrix} \partial_x+ i\partial_y & u \\ -\bar u  &   \partial_x- i\partial_y \end{bmatrix}  
    \begin{bmatrix} \psi_1  \\ \psi_2 \end{bmatrix}  = 0, 
\end{equation}
where 
\begin{equation}
\label{eq:DS_lax4}  
   \vec\psi = \begin{bmatrix} \psi_1  \\ i \psi_2 \end{bmatrix}. 
\end{equation}

By analogy with {\color{Violet} \cite{Taim1}, in some situations}  it is convenient to introduce the complex notation 
\begin{equation}
\label{eq:complex1}  
z= x+ i y, \ \ \bar z = x - i y, \ \  \partial_z = \frac{1}{2}(\partial_x-i\partial_y), \ \
\partial_{\bar z} = \frac{1}{2}(\partial_x+i \partial_y).
\end{equation}
To simplify our formulas we will write $u(z)$ instead of  $u(z,\bar z)$ without assuming that $u(z)$ is holomorphic. In the complex notation equation (\ref{eq:DS_lax3}) reads:
\begin{equation}
\label{eq:DS_lax5}
  \begin{bmatrix} 2\partial_{\bar z} & u(z) \\ -\bar u(z)  &   2\partial_{z}  \end{bmatrix}  
    \begin{bmatrix} \psi_1  \\ \psi_2 \end{bmatrix}  = 0, 
\end{equation}

\section{Finite-gap DS2 solutions}

In this Section we recall the construction of finite-gap at zero energy 2-D Dirac operators and the corresponding DS2 solutions. The finite-gap formulas for 2-D Dirac operators with the additional constraint $u=\bar u$ the finite-gap formulas were obtained in \cite{Taim3}, see also \cite{Taim6}.   

Let us point out that in the periodic theory of 2-D Dirac operator one can use different normalization of the wave function. In  \cite{Taim3,Taim5,Taim6} the following one was used:
\begin{equation}\label{eq:norm1}
\psi_1(\gamma,0) + \psi_2(\gamma,0) \equiv 1.
\end{equation}  
In the present text (analogously to \cite{GS1,GS5}) we work with the non-symmetric normalization for the wave function:
\begin{equation}\label{eq:norm2}
\psi_1(\gamma,0) \equiv 1.
\end{equation}
The main difference between these two normalizations is the following:
\begin{enumerate}
\item If the symmetric normalization (\ref{eq:norm1}) is used, the degree of the divisor is $g+1$, where $g$ is the genus of the spectral curve, and for small perturbations of the constant potential al least one of the divisor points is located far from the resonant point. If one uses  (\ref{eq:norm2}) instead, then the degree of the divisor is equal to the genus of the spectral curve, and for a small perturbation of the constant potential all divisor points are located near the resonant points.
\item If normalization (\ref{eq:norm1}) is used, the spectral data completely determines the potential $u(z)$. If one uses   (\ref{eq:norm2}), then the potential is determined by the spectral data only up to a constant phase factor, which should be added as additional parameter. As we pointed out above, the DS2 solutions are defined up to a phase factor, which is an arbitrary function of time, and the non-symmetric normalization  (\ref{eq:norm2}) ``hides'' this gauge freedom.
\end{enumerate}

The construction of finite-gap solutions for the DS2 equation consists of two steps:
\begin{enumerate}
\item Starting from a spectral curve with marked points and divisor, the ``complex'' Dirac operator
\begin{equation}
\label{eq:DS_lax6}
  \begin{bmatrix} 2\partial_{\bar z} & u(z) \\ v(z)  &   2\partial_{z}  \end{bmatrix}  
  \begin{bmatrix} \psi_1  \\ \psi_2 \end{bmatrix}  = 0,
\end{equation}
is constructed. Here ``complex'' means that the functions $u(z)$, $v(z)$ are independent. 
\item Additional conditions on the spectral curve and the divisor implying $v(z) = - \bar u(z)$ (or $v(z) = \bar u(z)$ in the defocusing case) are imposed. It is easy to check that these conditions are invariant with respect to the time evolution.
\end{enumerate}  

\subsection{Step 1. Complex Dirac operators.}

Consider the following collection of spectral data:
\begin{enumerate}
\item An algebraic Riemann surface $\Gamma$ of genus $g$ with two marked points $\infty_1$, $\infty_2$;
\item Local parameters $1/\lambda_1$, $1/\lambda_2$, near these  marked points;
\item A generic divisor ${\cal D} = \gamma_1+\ldots+\gamma_g$ of degree $g$ on $\Gamma$.  
\end{enumerate}
Then for generic data there exists an unique pair of functions $\psi_1(\gamma,z)$, $\psi_2(\gamma,z)$ such that:
\begin{enumerate}
\item They are holomorphis on $\Gamma$ outside the points $\infty_1$, $\infty_2$, $\gamma_1$, \ldots, $\gamma_g$;
\item They have first-order poles at the divisor points  $\gamma_1$, \ldots, $\gamma_g$;
\item They have the following essential singularities at the points  $\infty_1$, $\infty_2$:
\begin{equation}
  \label{eq:ba_1}
  \begin{bmatrix} \psi_1(\gamma,z)  \\ \psi_2(\gamma,z) \end{bmatrix}  =
  \begin{bmatrix} 1 + \frac {\xi_1^+(z)}{\lambda_1} +  \frac {\xi_2^+(z)}{\lambda_1^2} + \ldots) \\ \frac {\xi_1^-(z)}{\lambda_1} +  \frac {\xi_2^-(z)}{\lambda_1^2}+\ldots
  \end{bmatrix} e^{\lambda_1 z}, \ \ \mbox{as} \ \ \gamma\rightarrow\infty_1,
\end{equation}
\begin{equation}
  \label{eq:ba_2}
  \begin{bmatrix} \psi_1(\gamma,z)  \\ \psi_2(\gamma,z) \end{bmatrix}  =
  \begin{bmatrix} \chi_0^+(z) + \frac {\chi_1^+(z)}{\lambda_2} + \ldots) \\ X_{-1}\lambda_2+  \chi_0^-(z)+ \frac {\chi_1^-(z)}{\lambda_2} +\ldots
  \end{bmatrix} e^{\lambda_2 \bar z}, \ \ \mbox{as} \ \ \gamma\rightarrow\infty_2.
\end{equation}
In (\ref{eq:ba_1}), (\ref{eq:ba_2}) we assume that the pre-exponential terms are locally meromorphic (holomorphic) near the points $\infty_1$, $\infty_2$, $X_{-1}$ is a fixed constant, the expansion coefficients $\xi_j^+(z)$, $\xi_j^-(z)$,  $\chi_j^+(z)$, $\chi_j^-(z)$, are a priori unknown.
\end{enumerate}

Then using standard arguments one proves that $\psi_1(\gamma,z)$, $\psi_2(\gamma,z)$ satisfies (\ref{eq:DS_lax6}) with
\begin{equation}
  \label{eq:ba_3}
  u(z) =- 2  \chi_0^+(z)/ X_{-1},\ \  v(z) = -2 \xi_1^-(z).
\end{equation}

\subsection{Real reductions of DS2}

Assume now that the finite-gap spectral data satisfy the following additional reality conditions:
\begin{enumerate}
\item The spectral curve $\Gamma$ admits an antiholomorphic involution $\sigma:\Gamma\rightarrow\Gamma$, $\sigma^2=\mbox{id}$ such that:
\begin{itemize}
\item It interchanges the marked points: $\sigma\infty_j=\infty_{3-j}$, $j=1,2$;
\item $\sigma\lambda_1=\bar\lambda_2$, $\sigma\lambda_2=\bar\lambda_1$;  
\end{itemize}  
\item The divisor ${\cal D}$ satisfies the following reality condition:
  \begin{equation}
  \label{eq:cher0}
  {\cal D} - \sigma{\cal D} + \infty_2-\infty_1=0.
\end{equation} 
Equivalently, condition (\ref{eq:cher0}) means that there exists a function $f(\gamma)$ on $\Gamma$ such that
\begin{itemize}\label{item:funct_f}
\item $f(\gamma)$ has a first-order zero at $\infty_1$ and a first-order pole at $\infty_2$;
\item $f(\gamma)$ has first-order poles at ${\cal D}$ and first-order zeroes at $\sigma{\cal D}$. 
\end{itemize}
\end{enumerate}
\begin{lemma}
  Let $f(\gamma)$ be the function with the poles at  ${\cal D}$,  $\infty_2$ and zeroes at $\sigma{\cal D}$,  $\infty_1$ with the normalization:
 \begin{equation} \label{eq:cher01}
   f(\gamma) = \lambda_2 + O(1) \ \ \mbox{near the point} \ \ \infty_2.   
\end{equation}   
Denote by $\cal C$ the coefficient at the leading term of the expansion for $f(\gamma)$ near $\infty_1$:
 \begin{equation} \label{eq:cher01}
   f(\gamma) = \frac{\cal C}{\lambda_1} + O\left(\frac{1}{\lambda_1^2}\right) \ \ \mbox{near the point} \ \ \infty_1.   
\end{equation}  
Then for generic data the constant $\cal C$ is real.
\end{lemma}
\begin{proof}
  Consider the function
\begin{equation} \label{eq:cher02}
   f_1(\gamma) = \overline{\left(\frac{\cal C}{f(\sigma\gamma)} \right)}.   
\end{equation}  
It is holomorphic, has the same poles and zeroes as $f(\gamma)$ and 
  \begin{equation} \label{eq:cher03}
   f_1(\gamma) = \lambda_2 + O(1) \ \ \mbox{near} \ \ \infty_2, \ \  f_1(\gamma) = \frac{\bar{\cal C}}{\lambda_1} + O\left(\frac{1}{\lambda_1^2}\right) \ \ \mbox{near} \ \ \infty_1.   
 \end{equation}
For generic data the function $f(\gamma)$ is unique, therefore $f_1(\gamma)=f(\gamma)$, and $\bar{\cal C}= {\cal C}$. 
\end{proof}

\begin{lemma}
  Let the spectral data satisfy the reality conditions formulated in this Section, and the normalization constant $X_{-1}$ satisfies:
 \begin{equation}
  \label{eq:cher3}
  |X_{-1}|^2=-\eta {\cal C}^{-1}.
\end{equation} 
Then the potentials $u(z)$, $v(z)$  in (\ref{eq:DS_lax6}) satisfy the DS2 reduction
\begin{equation}
  \label{eq:cher4}
v(z)=-\eta\bar u(z).
\end{equation}
\end{lemma}
\begin{proof}
  If the spectral data satisfies the reality conditions, then assume
\begin{equation}
  \label{eq:cher4}
  \psi_2(\gamma,z) = X_{-1} f(\gamma)\overline{\psi_1(\sigma\gamma,z)}.
\end{equation}  
The function $\psi_2(\gamma,z)$ defined by (\ref{eq:cher4}) has the correct analytic properties, and
\begin{equation}
  \label{eq:cher5}
  \xi_1^+(z) = X_{-1} X{\cal C} \overline{\chi_0^+(z)}  = -\frac{\overline{\chi_0^+(z)}}{\eta \overline{ X_{-1}}}.
\end{equation}  
Taking into account that 
\begin{equation}
  \label{eq:cher5}
 u(z) =- 2 \frac{\chi_0^+(z)}{X_{-1}},\ \  v(z) = -2 \xi_1^-(z)=  2 \frac{ \overline{\chi_0^+(z)}}{\eta \overline{ X_{-1}}},
\end{equation}
we complete the proof.
\end{proof}
Let us remark that for $z=0$ we have:
$$
\psi_1(\gamma,0)\equiv 1, \ \ \chi_0^+(0)=1,
$$
therefore
\begin{equation}\label{eq:xi-1}
X_{-1}=- 2 / u(0).
\end{equation}
{\color{Violet}
\begin{remark}
The antiholomorphic involution $\sigma$ in our paper coincides with the product of the holomorphic involution $\sigma$ and antiholomorphic involution $\tau$ in \cite{Taim3}. Therefore the reality condition (\ref{eq:cher0}) on the divisor coincides with the consequence of the corresponding pair of constraints on the divisor associated with these involutions in \cite{Taim3}.
  
\end{remark}  
}  
\begin{remark}
  In the finite-gap theory of soliton equations we often meet the following situation. It is sufficiently easy to construct solutions of a generalization of the original system, but the selection of the spectral data corresponding to the solutions of original system requires additional efforts.

  One of the problems of such type is the selection of real solutions (or real regular ones). It is well-known that, for the Korteweg-de Vries equation, defocusing Nonlinear Schr\"odinger equation, $\sinh$-Gordon equation, Kadomtsev-Petviashvili 2 equation, the real solutions correspond to complex curves with antiholomorphic involution (complex conjugation), and divisors invariant with respect to this complex conjugation. But, for equations like self-focusing Nonlinear Schr\"odinger equation,  $\sin$-Gordon equation, Kadomtsev-Petviashvili 1 equation, the characterization of divisors corresponding to real soltuions is much less trivial. Again the spectral curve admits a complex conjugation, but the conditions on the divisor are much less trivial. For self-focusing Nonlinear Schr\"odinger equation and  $\sin$-Gordon equation these conditions were first found in \cite{Cher}, and the answer was formulated in terms of the some meromorphic differential, now known as Cherednik differential. 

{\color{Violet} Let us remark that the condition (\ref{eq:cher0}) is neither a ``naive'' reality condition analogous to the defocusing NLS, nor a Cherednik type condition, because  (\ref{eq:cher0}) does not include the canonical class of the surface. 
}

  Another problem of such type is the selection of solutions such that some of the fields are identically zero. An important problem of this type is the fixed-energy periodic problem for the 2-dimensional stationary Schr\"odinger operator. Finite-gap at one energy operators with non-zero magnetic field were constructed in 1976 in \cite{DKN}, but the selection of spectral data corresponding to operators with identically zero magnetic field took almost 7 years; and the solution, obtained in \cite{VN1,VN2}, again was formulated in terms of a Cherednik differential. Selection of spectral data generating pure magnetic operators was obtained much later in \cite{GMN}, and it used a completely different approach.
\end{remark}

\section{Theta-functional formulas}

Denote by $\omega_j$ the basis of holomorphic differentials:
\begin{equation}\label{eq:omegas1}
\oint\limits_{a_j} \omega_k= 2\pi i \delta_{jk}, \ \ \oint\limits_{b_j} \omega_k= b_{jk}, 
\end{equation}
where $b_{jk}$ denotes the Riemann matrix of periods. Let us recall that this matrix is symmetric $b_{jk}=b_{kj}$ and its real part is negatively defined.
\begin{remark} Let us recall that there are two standard normalizations of basic holomorphic differentials in the literature: either (\ref{eq:omegas1}) (see  \cite{BBEIM,Dubrovin2,Fay} )  or (\ref{eq:omegas1bis}) (see \cite{Mum1983} )
\begin{equation}\label{eq:omegas1bis}
\oint\limits_{a_j} \omega_k= \delta_{jk},
\end{equation}
and before using the theta-functional formulas it is important to check which of these normalizations is used.
\end{remark}

We also need the following meromorphic differentials $\Omega_0$, $\Omega_z$,  $\Omega_{\bar z}$,  $\Omega_t$ such that:
\begin{enumerate}
\item They are holomorphic outside the marked points $\infty_1$,  $\infty_2$;
\item They have the following asymptotics near the marked points:   
  \begin{align}\label{eq:omegas2}
    \Omega_0 &= \left\{\begin{array}{ll} \left[-\frac{1}{\lambda_1} + O\left(\frac{1}{\lambda_1^2} \right) \right] d\lambda_1  & \mbox{  at  } \infty_1, \\ \\ \left[\frac{1}{\lambda_2} + O\left(\frac{1}{\lambda_2^2} \right) \right] d\lambda_2 & \mbox{  at  } \infty_2,  \end{array}   \right.\\
    \Omega_z &= \left\{\begin{array}{ll} \left[1 + O\left(\frac{1}{\lambda_1^2} \right) \right]d\lambda_1  & \mbox{  at  } \infty_1, \\ O\left(\frac{1}{\lambda_2^2} \right) d\lambda_2 & \mbox{  at  } \infty_2,  \end{array}   \right.\\
   \Omega_{\bar z} &= \left\{\begin{array}{ll} O\left(\frac{1}{\lambda_1^2} \right) d\lambda_1  & \mbox{  at  } \infty_1,  \\ \left[1 + O\left(\frac{1}{\lambda_2^2} \right) \right] d\lambda_2 & \mbox{  at  } \infty_2,  \end{array}   \right.\\
    \Omega_t &= \left\{\begin{array}{ll} \left[ 4i \lambda_1 + O\left(\frac{1}{\lambda_1^2} \right) \right] d\lambda_1  & \mbox{  at  } \infty_1, \\ \\ \left[ -4i \lambda_2 + O\left(\frac{1}{\lambda_2^2} \right) \right] d\lambda_2 & \mbox{  at  } \infty_2;  \end{array}   \right.
\end{align}
\item All $a$-periods of $\Omega_0$, $\Omega_z$,  $\Omega_{\bar z}$,  $\Omega_t$ are equal to zero:
\begin{equation}\label{eq:omegas3}
\oint\limits_{a_j} \Omega_0= \oint\limits_{a_j} \Omega_z= \oint\limits_{a_j} \Omega_{\bar z}=\oint\limits_{a_j} \Omega_{t}= 0, \ \ j=1,\ldots,g. 
\end{equation}  
\end{enumerate}
Without loss of generality we may assume that the starting point of the Abel transform coincides with $\infty_1$, and we fixe a path  $\cal P_0$ connecting $\infty_1$ with $\infty_2$. Consider the antiderivatives of these differentials in a neighbourhood of this path
$$
\Omega_0=dF_0, \ \ \Omega_z=dF_z, \ \ \Omega_{\bar z}=dF_{\bar z}, \ \ \Omega_t=dF_t,
$$
and assume that the integration constants are chosen by assuming that near $\infty_1$:
  \begin{equation}\label{eq:omegas4}
    F_0 = -\log{\lambda_1} + o(1), \ \  F_z = \lambda_1 +o(1), \ \ F_{\bar z} = o(1), \ \ 
    F_t = 2i \lambda_1^2 +  o(1).
\end{equation}
Then the constants ${\cal C}_0$,  ${\cal C}_{z}$,  ${\cal C}_{\bar z}$,  ${\cal C}_{t}$ are defined by the expansions of these aniderivatives near  $\infty_2$: 
\begin{equation}\label{eq:omegas5}
    F_0 = \log{\lambda_2} + {\cal C}_0+ o(1), \  F_z = {\cal C}_{z} +o(1), \ F_{\bar z} =\lambda_2+{\cal C}_{\bar z}+ o(1), \ F_t = -2i \lambda_2^2 + {\cal C}_{t} + o(1).
\end{equation}

The eigenfunctions of the zero-curvature representation are given by the Its formula, {\color{Violet} provided for the 2-d Dirac operator in \cite{Taim3})}:
\begin{align}
  &\psi_1(\gamma,z,t)=\exp\left[z\int^{\gamma}\Omega_z+ \bar z\int^{\gamma}\Omega_{\bar z}+ t\int^{\gamma}\Omega_t \right]\times\label{eq:psi1} \\ & \times \frac{\theta(\vec A(\gamma) + \vec W_z z + \vec W_{\bar z} \bar z + \vec W_{t} t  - \vec A({\cal D}) -\vec K)\,\theta( - \vec A({\cal D}) -\vec K) }   {\theta(\vec A(\gamma) - \vec A({\cal D}) -\vec K)\,\theta( \vec W_z z + \vec W_{\bar z} \bar z + \vec W_{t} t  - \vec A({\cal D}) -\vec K)}\nonumber
  \\
  & \psi_2(\gamma,z,t) = X_{-1} \, f(\gamma)\, \overline{\psi_1(\gamma,z,t)}, \ \ \mbox{where},\\
  & f(\gamma) = {\cal C}_0^{-1} \exp\left[\int^{\gamma}\Omega_0\right]  \frac{\theta(\vec A(\gamma) + \vec W_0 - \vec A({\cal D}) -\vec K)\,\theta(\vec A(\infty_2) - \vec A({\cal D}) -\vec K)}   {\theta(\vec A(\gamma) - \vec A({\cal D}) -\vec K)\,\theta(\vec A(\infty_2) + \vec W_0 - \vec A({\cal D}) -\vec K)}.\label{eq:psi2}
\end{align}
Here we assume that the integration constants for integrals in (\ref{eq:psi1}), (\ref{eq:psi2}) are chosen as in (\ref{eq:omegas4}). We also assume that in  (\ref{eq:psi2}) the path connecting $\infty_1$ with $\gamma$ is the path $\cal P_0$ plus a path connecting $\infty_2$ with $\gamma$.

Expanding (\ref{eq:psi1}) near $\infty_2$ we obtain:
\begin{align}
  &\chi_0^+(z,t) = \exp\left[z {\cal C}_z + \bar z {\cal C}_{\bar z} + t {\cal C}_t \right]\times \label{eq:psi3} \\
  & \times \frac{\theta(\vec A(\infty_2) + \vec W_z z + \vec W_{\bar z} \bar z + \vec W_{t} t  - \vec A({\cal D}) -\vec K)\,\theta( - \vec A({\cal D}) -\vec K) }   {\theta(\vec A(\infty_2) - \vec A({\cal D}) -\vec K)\,\theta( \vec W_z z + \vec W_{\bar z} \bar z + \vec W_{t} t  - \vec A({\cal D}) -\vec K)},
\end{align}
and, then, taking into account (\ref{eq:xi-1}),
\begin{align}
  &u(z,t) = \exp\left[z {\cal C}_z + \bar z {\cal C}_{\bar z} + t {\cal C}_t \right]\times \label{eq:psi3} \\
  & \times \frac{\theta(\vec A(\infty_2) + \vec W_z z + \vec W_{\bar z} \bar z + \vec W_{t} t  - \vec A({\cal D}) -\vec K)\,\theta( - \vec A({\cal D}) -\vec K) }   {\theta(\vec A(\infty_2) - \vec A({\cal D}) -\vec K)\,\theta( \vec W_z z + \vec W_{\bar z} \bar z + \vec W_{t} t  - \vec A({\cal D}) -\vec K)} u(0,0).\nonumber
\end{align}
{\color{Violet}
Let us recall that the Riemann theta-function is defined as the following Fourier series \cite{BBEIM,Dubrovin2,Fay}:
\begin{equation}
\label{eq:theta0}
\theta(\vec z)=\theta(\vec z|B)=\sum\limits_{\begin{array}{cc}n_j\in\ZZ \\ j=1,\ldots,g\end{array}}
\exp{\left[ \frac{1}{2} \sum_{j,k=1}^g b_{jk}n_j n_k +  \sum_{j=1}^g n_j z_j  \right] },
\end{equation}
$$
\vec z = (z_1,\ldots,z_g) \in \CC^g.
$$
}  

\section{Direct spectral transform}
For the spatially one-dimensional problems the existence of the spectral curve for a generic periodic potential can be proved rather easily using the classical theory of ordinary differential equations. In contrast with the one-domensional case, the two-dimensional requires much more serious analytic tools. Currently, two main approaches are used. In \cite{Taim2,Taim6} the spectral curve for the doubly-periodic problem (\ref{eq:DS_lax5}) was constructed using the Keldysh theorem for holomorphic Fredholm operator pencils. In \cite{Krichever2,Krichever3} the spectral curve for the non-stationary heat conductivity operator and for two-dimensional doubly-periodic Schr\"odinger operator at a fixed energy level was constructed using the perturbation theory, and in our paper we follow this approach.
\subsection{Unperturbed spectral curve}

Let us remark that, in contrast with the one-dimensional case, the existence of a ``good'' spectral curve essentially depends on the analytic properties of the operators. For example, for the Lax operators associated with DS1 equation of Kadomtsev-Petviashvily 1 equation, the curve constructed using the aforementioned methods does not admit good local compactification \cite{Taim6}. 

Using the scaling properties of the DS equation, we may assume without loss of generality that $a=1$. Denote by $\cal L$ the Dirac operator multiplied by $\sigma_3$
\begin{equation}
\label{eq:DS_lax4}
 {\cal L}= {\cal L}_0 + \epsilon  {\cal L}_1, \ \  {\cal L}_0 = \begin{bmatrix} \partial_x+ i\partial_y & 1 \\ 1 &   - \partial_x+ i\partial_y \end{bmatrix}, \ \ {\cal L}_1 = \begin{bmatrix} 0  & v(x,y) \\ \bar v(x,y)  &  0 \end{bmatrix}.
\end{equation}
The zero Bloch eigenfunctions for ${\cal L}_0$ are:
\begin{equation}\label{eq:DS_lax4.1}
\Psi_0(p,x,y,t) = \begin{bmatrix} 1 \\ - ip +q \end{bmatrix} \exp\left[i\left(px+ q y\right)-2pqt \right], \ \ \mbox{with} \ \ p^2+q^2=1, \ \ p,q\in\CC, 
\end{equation}
therefore the unperturbed spectral curve $\Gamma_0$ is defined by:
\begin{equation}\label{eq:DS_lax4.2}
p^2+q^2=1,
\end{equation}
and it can be conveniently parameterized as
\begin{equation}\label{eq:param}
p = \frac{1}{2}\left[\tau +\frac{1}{\tau} \right], \ \ q= -\frac{i}{2}\left[\tau -\frac{1}{\tau} \right],
\end{equation}
i.e.
\begin{equation} 
  \tau = p + i q, \ \  \frac{1}{\tau} = p - iq.
\end{equation}
The unperturbed wave function can also be written as:
\begin{equation}
\Psi_0(\tau,z,t) = \begin{bmatrix} 1 \\ -i\tau \end{bmatrix} \exp\left[\frac{i}{2}\left(\tau\bar z + \frac{z}{\tau}+\left(\tau^2-\frac{1}{\tau^2} \right)t  \right) \right].
\end{equation}

The marked points and the local parameters are respectively:
\begin{equation}
  \infty_1: \ \tau=0, \ \ \infty_2: \ \tau=\infty, \ \ \lambda_1=\frac{i}{2\tau}, \ \ \lambda_2=\frac{i\tau}{2}.
\end{equation}
Eigenfunctions (\ref{eq:DS_lax4.1}) are Bloch periodic
\begin{equation}\label{eq:DS_lax4.3}
\Psi_0(p,x+L_x,y,t)= \varkappa_x \Psi_0(p,x,y,t), \  \ \Psi_0(p,x,y+L_y,t)= \varkappa_y \Psi_0(p,x,y,t), 
\end{equation}
with the Bloch multipliers
\begin{equation}\label{eq:DS_lax4.4}
\varkappa_x = \exp\left[ ip L_x \right], \ \ \varkappa_y = \exp\left[ iq L_y \right].
\end{equation}
Following \cite{Krichever2,Krichever3} consider the image of $\Gamma_0$ under the map (\ref{eq:DS_lax4.4}). A pair of points $\tau_1=p_1+iq_1$,  $\tau_2=p_2+iq_2$ is called \textbf{resonant} if the images of these points coincide, namely
\begin{equation}\label{eq:DS_lax4.5}
\varkappa_x(\tau_1)  =  \varkappa_x(\tau_2), \ \ \varkappa_y(\tau_1)  =  \varkappa_y(\tau_2).
\end{equation}
Doubly-periodic small perturbations of operators results in transformation of double points into thin handles \cite{Krichever2,Krichever3,Taim4,Taim6}, therefore it is natural to develop perturbation theory near such pairs.  

By analogy with \cite{GS5}, using the scaling symmetry of DS2, we assume that
\begin{equation}\label{eq:zero_harm}
\int\limits_{0}^{L_x}\int\limits_{0}^{L_y}  v_0(x,y) dx dy =0.     
\end{equation}
This assumption essentially simplifies the calculations.

The Fourier harmonics of the perturbation are enumerated by pairs of integers:
\begin{equation}\label{eq:one_harm}
k_x = n_x \frac {2\pi}{L_x}, \ \  k_y = n_x \frac {2\pi}{L_y}, \ \ n_x,n_x\in\ZZ.  
\end{equation}
In our paper we assume that the periods $L_x$, $L_y$ are generic, therefore
\begin{enumerate} 
 \item $k_x^2+k_y^2\ne 4$ for all $n_x$, $n_y$;  
 \item All multiple points of the image of (\ref{eq:DS_lax4.4}) are double points.
\end{enumerate}
For non-generic periods one has to study perturbations of higher-order multiple points. This problem may be very interesting, but it requires a serious additional investigation, therefore we do not discuss it now.  

Equations (\ref{eq:DS_lax4.5}) are equivalent to the following pair of equations:
\begin{equation}\label{eq:DS_lax4.6}
  \left\{\begin{array}{l} \tau_2-\tau_1 = k_x+ i k_y, \\ \frac{1}{\tau_2}-\frac{1}{\tau_1} = k_x- i k_y,
         \end{array} \right.
\end{equation}
where $k_x$, $k_y$ are defined by (\ref{eq:one_harm}) with integer $n_x,n_y$. 

\begin{remark}Equation (\ref{eq:DS_lax4.6}) has the following interpretation: the matrix elements for the harmonic perturbation with a given $n_x,n_y$ are non-zero only if (\ref{eq:DS_lax4.6}) is fulfilled. Therefore in the leading order of perturbation theory only the wave functions of corresponding resonant pairs appear. 
\end{remark}

If $k_x^2+k_y^2<4$, this mode is unstable, otherwise it is stable. We have two types of resonant pairs $(\tau_1,\tau_2)$ respectively:
\begin{enumerate}
\item Resonant pairs corresponding to the unstable modes $k_x^2+k_y^2<4$:
\begin{align}\label{eq:res_unstable}
  \tau_1&= \frac{k_x+ i k_y}{2}\left[-1 \pm i\sqrt{\frac{4- k_x^2 -k_y^2}{k_x^2+k_y^2}}\right], \\
  \tau_2&= \frac{k_x+ i k_y}{2}\left[1 \pm i\sqrt{\frac{4- k_x^2 -k_y^2}{k_x^2+k_y^2}}\right], \ \ |\tau_1|=|\tau_2|=1
          \nonumber
\end{align}
In analogy with the NLS case \cite{GS1,GS5}, it is convenient to parametrize the unstable modes by angles:
\begin{equation}
  k_x = 2 \cos\phi \cos\theta,\ \  k_y = 2 \cos\phi \sin\theta. 
\end{equation}  
Then
\begin{equation}
\tau_1 = -e^{i(\theta\mp\phi)}, \ \  \tau_2  = e^{i(\theta\pm\phi)}.
\end{equation}  

\item Resonant pairs corresponding to the stable modes $k_x^2+k_y^2>4$ ;
\begin{equation}\label{eq:res_stable}
\tau_1= \frac{k_x+ i k_y}{2}\left[-1+ \sqrt{\frac{k_x^2+k_y^2-4}{k_x^2+k_y^2}}\right], \ \ \tau_2 = -\frac{1}{\bar\tau_1}.
\end{equation}
\end{enumerate}

\begin{remark} Due to the reality condition, the wave vectors $(k_x,k_y)$ and  $(-k_x,-k_y)$ appear simultaneously, and they correspond to the same unstable mode.  
\end{remark}
  
Following \cite{GS1} we introduce a finite-gap approximation by neglecting all stable modes. Consider all resonant pairs $(\tau_{2j-1},\tau_{2j})$ corresponding to the unstable modes, $j=1,\ldots,2N$, where $N$ is the number of unstable modes. We also assume that $\Im\tau_{2j-1}\tau_{2j}^{-1}>0$ for all $j$ and the points $\tau_1$,  $\tau_3$,  $\tau_5$, \ldots, $\tau_{2j-1}$ are ordered clockwise.

\begin{example}\label{ex:ex1}
  Let $L_x=2\pi/1.2$, $L_y=2\pi/1.4$. Then $k_x=1.2 n_x$, $k_y=1.4 n_y$ and we have 4 unstable modes
\begin{equation}
(n_x,n_y) = (1,0), \ \ (n_x,n_y) = (0,1), \ \ (n_x,n_y) = (1,1), \ \ (n_x,n_y) = (1,-1). 
\end{equation}
Therefore we have 8 pairs of resonant points:
\begin{itemize}
\item $(\tau_1,\tau_2)$ and  $(\tau_{9},\tau_{10})$ correspond to $(n_x,n_y) = (1,1)$;
\item $(\tau_3,\tau_4)$ and  $(\tau_{11},\tau_{12})$ correspond to $(n_x,n_y) = (1,0)$;
\item $(\tau_5,\tau_6)$ and  $(\tau_{13},\tau_{14})$ correspond to $(n_x,n_y) = (1,-1)$;  
\item $(\tau_7,\tau_8)$ and  $(\tau_{15},\tau_{16})$ correspond to $(n_x,n_y) = (0,1)$. 
\end{itemize}  
\begin{figure}[H]
  \centering{
    \includegraphics[width=0.3\textwidth]{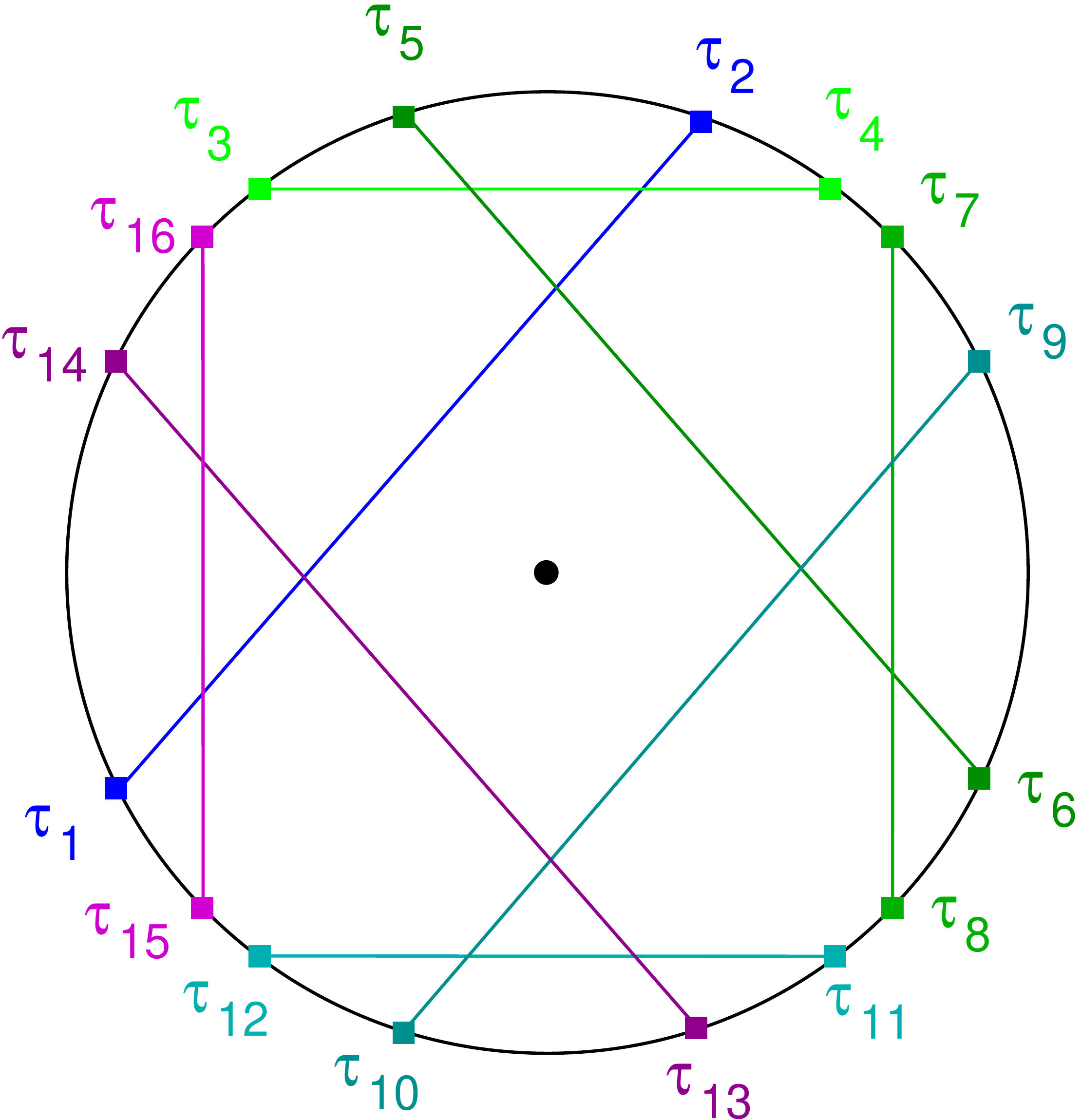}\hspace{5mm}
    \includegraphics[width=0.3\textwidth]{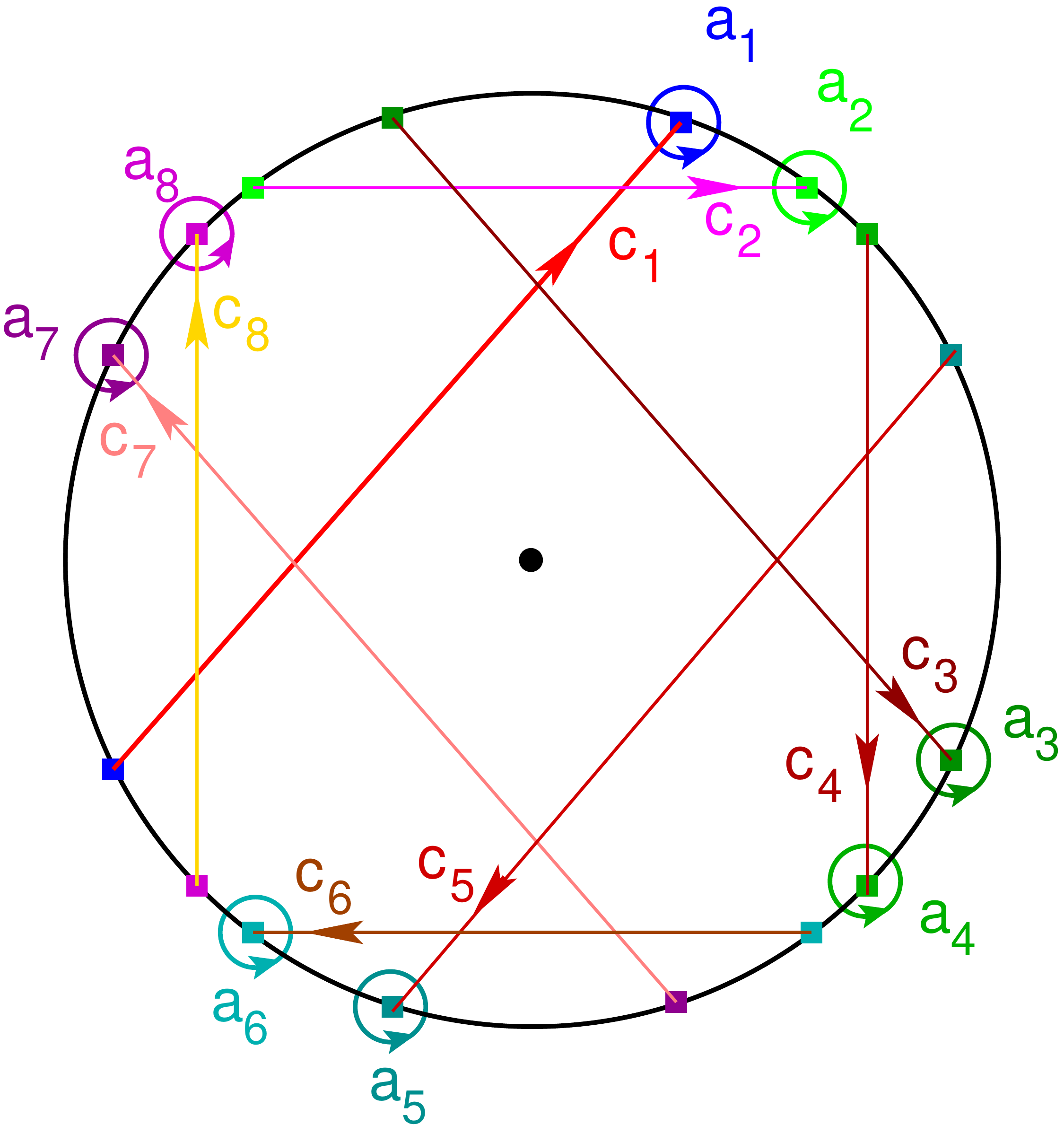}\hspace{5mm}
    \includegraphics[width=0.3\textwidth]{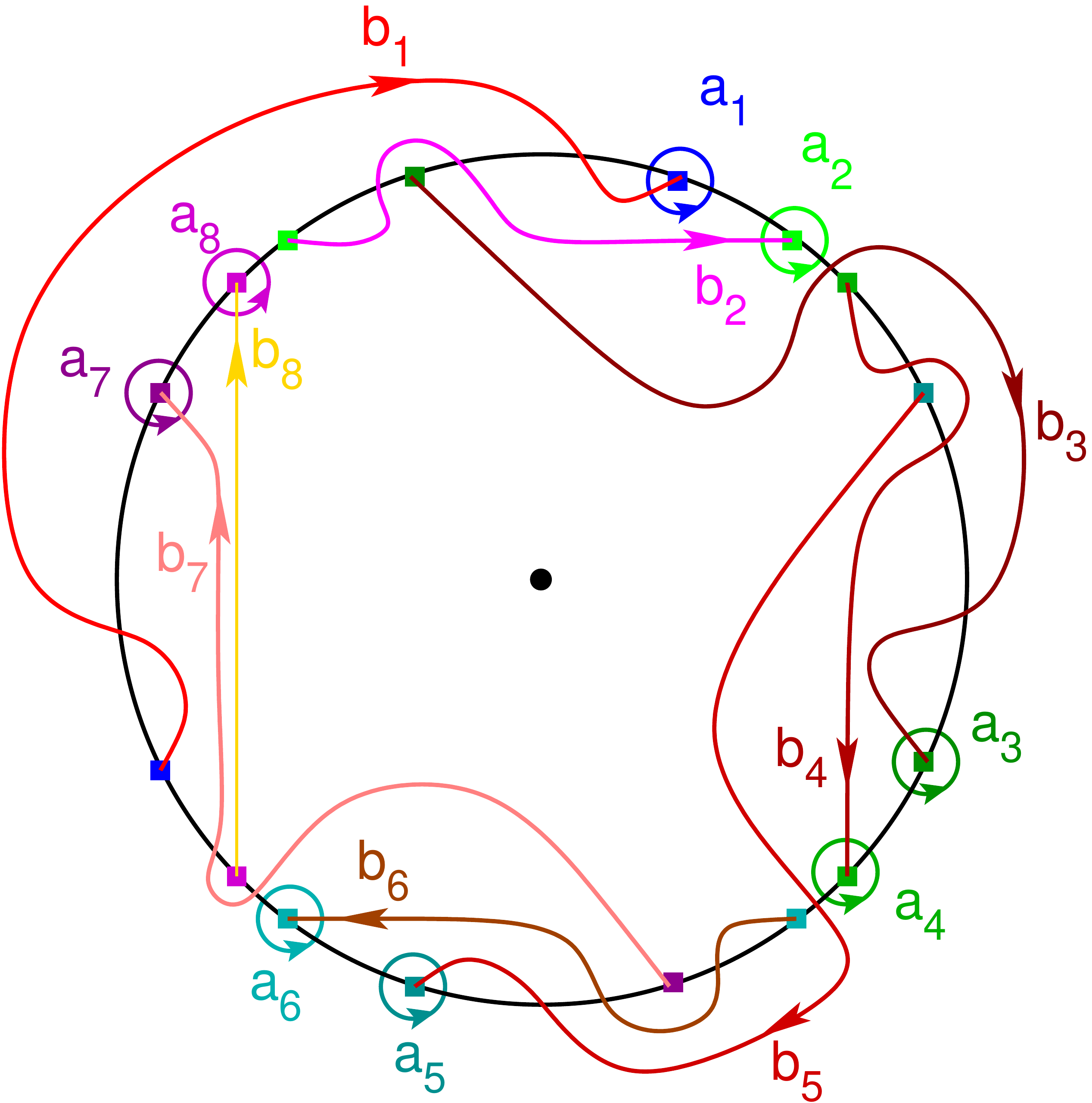}
      \caption{\small{\sl On the left: enumeration of the resonant points for Example~\ref{ex:ex1}. In the middle: the system of $a$ and $c$-cycles for Example~~\ref{ex:ex1}. On the right: the corresponding system of $a$ and $b$ cycles.}}\label{fig:enumeration}}
\end{figure}
 \end{example}

Let us introduce a system of basic cycles for the unperturbed spectral curve. Let $a_j$ be a small cycle about the point $\tau_{2j}$ oriented counterclockwise, or equivalently,  a small cycle about the point $\tau_{2j-1}$ oriented clockwise. Denote by $c_j$ the line, starting at the point $ \tau_{2j-1}$ and ending at the point $\tau_{2j}$. Of course
$$
a_j \cdot c_k = \delta_{jk},
$$
but $c_j\cdot c_k$ is not necessary 0. But if we define $b_j$ by 
\begin{equation} \label{eq:b-cycles}
b_j = c_j-\sum_{k>j} (c_j\cdot c_k) c_k,
\end{equation}
we obtain a canonical basis of cycles on the unperturbed curve:
$$
a_j \cdot a_k = b_j \cdot b_k=0, \ \   a_j \cdot b_k  =  \delta_{jk}.
$$

\begin{remark}
The choice of $b$-cycles corresponding to a given system of $a$-cycles is not unique, but any integer symplectic transformation
\begin{equation}\label{eq:symplectic}
b_j\rightarrow b_j + \sum d_{ji} a_k, \ \ d_{jk}\in\ZZ, \ \ d_{jk}=d_{kj},
\end{equation}  
maps a system of $b$-cycles corresponding to the given system of $a$-cycles, to another one. But if all $d_{jj}=0 \ \ (\!\!\!\!\mod 2)$, this transformations of cycles does not affect the $\theta-function$. Therefore different selections of the systems of $b$-cycles generate the same $\theta$-function.
\end{remark}
\begin{figure}[H]
  \centering{\includegraphics[width=0.8\textwidth]{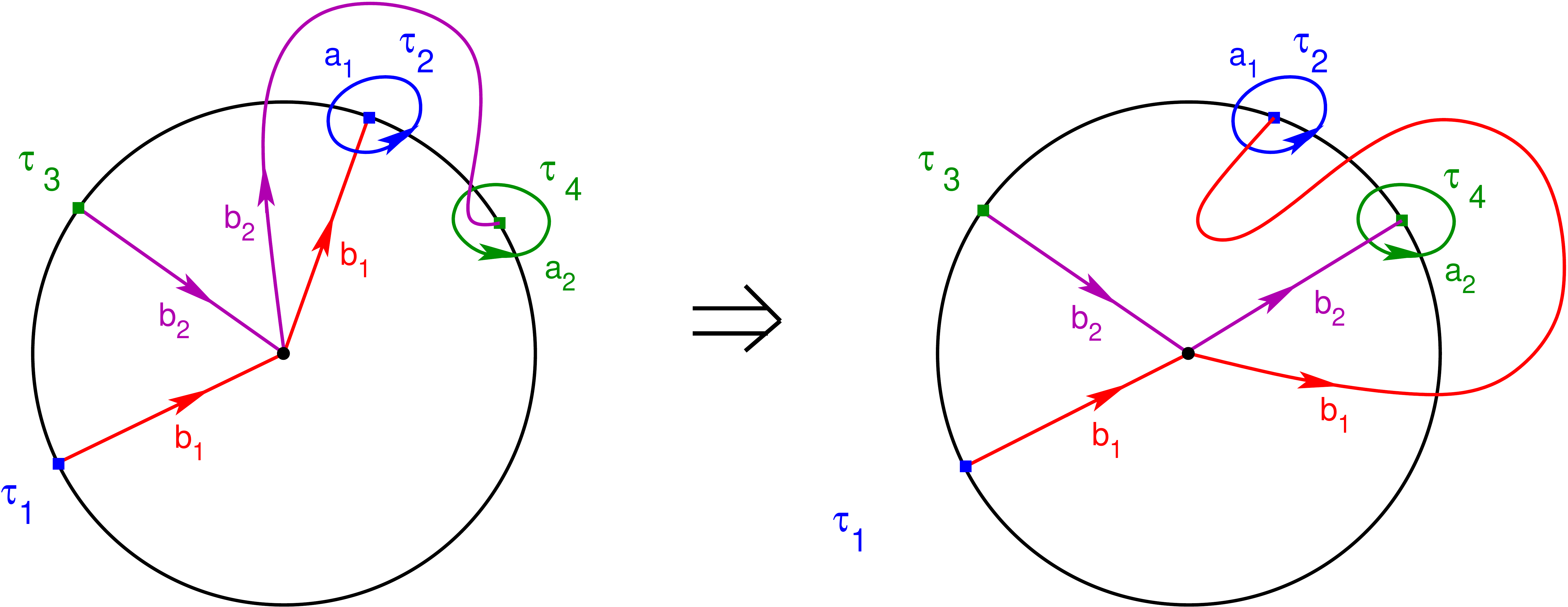}
      \caption{\small{\sl An elementary symplectic transformations of a system of $b$-cycles: $b_1\rightarrow b_1+a_2$, $b_2\rightarrow b_2+a_1$.}}\label{fig:enumeration}}
\end{figure}

In this paper we use the same approximation as in \cite{GS1}: the off-diagonal terms of the Riemann matrix, and the periods of the meromorphic differentials are calculated for the unperturbed curve.

The basic differentials on the unperturbed curve $\Gamma_0$ are:
\begin{align}
&\omega_j = \left[\frac{1}{\tau-\tau_{2j}}- \frac{1}{\tau-\tau_{2j-1}}\right]d\tau = d\log\left[\frac{\tau-\tau_{2j}}{\tau-\tau_{2j-1}}  \right], \ \ j=1,\ldots,g, \label{eq:unpert01}\\
  & dp = d\left(\frac{1}{2}\left[\tau +\frac{1}{\tau} \right]\right) =\frac{i q}{\tau}\,d\tau   , \ \ dq= -d\left(\frac{i}{2}\left[\tau -\frac{1}{\tau} \right]\right)=-\frac{i p}{\tau}\,d\tau, \label{eq:unpert02}\\
  & \Omega_0= \frac{d\tau}{\tau} =d\log(\tau), \ \ \Omega_{z} =-\frac{id\tau}{2\tau^2}= d\left(\frac{i}{2\tau}\right), \ \  \Omega_{\bar z} = \frac{i}{2} d\tau, \ \ \Omega_t =\frac{i}{2}d\left(\tau^2-\frac{1}{\tau^2} \right) 
\end{align}

\begin{lemma} For the unperturbed curve the periods of the differentials are the following:
  \begin{align}
    & \int\limits_{b_j} \Omega_0 = \log\left[\frac{\tau_{2j}}{\tau_{2j-1}} \right] = \log\left[\tau_{2j}\bar\tau_{2j-1}\right], \ \ \left(W_z\right)_j=\int\limits_{b_j} \Omega_z = \frac{i}{2}\left[\bar\tau_{2j}-\bar\tau_{2j-1} \right], \label{eq:abel01} \\
    &  \left(W_{\bar z}\right)_j=\int\limits_{b_j} \Omega_{\bar z} = \frac{i}{2}\left[ \tau_{2j}- \tau_{2j-1} \right],  \ \
      \left(W_t\right)_j =\int\limits_{b_j} \Omega_t = \Im(\tau_{2j-1}^2-\tau_{2j}^2), \label{eq:abel02} \\
    & A_j(\infty_2) - A_j(\infty_1) =  \log\left[\frac{\tau_{2j-1}}{\tau_{2j}} \right] =\log\left[\tau_{2j-1}\bar\tau_{2j}\right], \label{eq:abel03}\\
    & b_{jk} = \log\left[\frac{(\tau_{2j}-\tau_{2k}) (\tau_{2j-1}-\tau_{2k-1})}{(\tau_{2j}-\tau_{2k-1}) (\tau_{2j-1}-\tau_{2k}) } \right], \ \ k\ne j. \label{eq:abel04}\
  \end{align}  
  Moreover, for the unperturbed curve $\Gamma_0$ we have:
  \begin{equation}\label{eq:constants2}
   {\cal C}_0= {\cal C}_z= {\cal C}_{\bar z} = {\cal C}_t =0. 
  \end{equation}
\end{lemma}
\begin{remark}
The double ratio in (\ref{eq:abel04}) is always real, but its sign depends on the relative positions of the points $\tau_{2j-1}$, $\tau_{2j}$, $\tau_{2k-1}$, $\tau_{2k}$. The possible cases are
\begin{figure}[H]
  \centering{
    \includegraphics[width=0.7\textwidth]{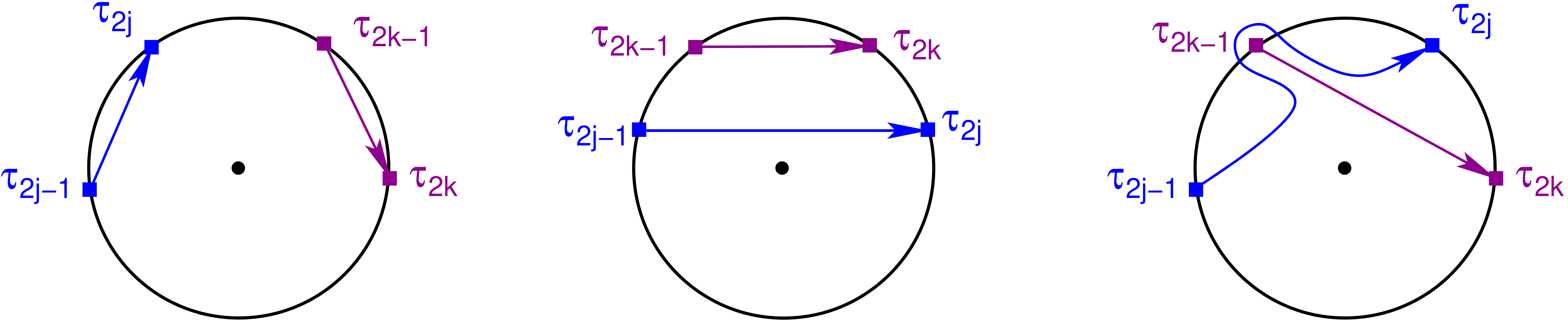}
    \caption{\small{\sl Possible configurations in  Formula (\ref{eq:abel04}). On the left and in the middle the double ratio is positive and the period $b_{jk}$ is real. On the right the double ratio is negative, and $b_{jk}$ is real number plus $\pi i$.}}\label{fig:periods1}}
\end{figure}
\end{remark}
\subsection{Perturbed spectral curve in the leading order}

Let us restrict ${\cal L}$ to the space of Bloch functions ${\cal F}(p,q)$  with fixed Bloch multipliers:
\begin{equation}
\label{eq:bloch1}
\begin{split}
\psi(x+L_x,y) = \varkappa_x \psi(x,y), \ \ \psi(x,y+L_y) = \varkappa_y \psi(x,y), \\ 
\varkappa_x = \exp{[iL_x p]}, \ \ \varkappa_y = \exp{[iL_yq]}, \ \ p,q\in\CC.
\end{split}
\end{equation}
Operator ${\cal L}_0$ has the following basis of eigenfunctions in ${\cal F}(p,q)$:
\begin{equation}
\label{eq:bloch2}
\psi^{(\pm)}_{m,n}=\begin{bmatrix}1 \\ -i p_m \pm \sqrt{1-p_m^2}  \end{bmatrix} \exp\big(i [p_m x  + q_n y] \big), 
\end{equation}
where
\begin{equation}
\label{eq:bloch21}
p_m=p+ \frac{2\pi}{L_x}m, \ \  q_n = q+ \frac{2\pi}{L_y}n,  \ \ m,n\in\ZZ,
\end{equation}
solving the eigenvalue equation
\begin{equation}
\label{eq:bloch22}
{\cal L}_0 \psi^{(\pm)}_{m,n} =\left(-q_n \pm  \sqrt{1-p_m^2}\right)  \psi^{(\pm)}_{m,n}.
\end{equation}
In (\ref{eq:bloch1})-(\ref{eq:bloch22}) we do not assume that $p^2+q^2=1$.

Consider a monochromatic unstable perturbation
\begin{equation}\label{eq:pert1}  
v(x,y) = c_{j} e^{i(k_x x + k_y y)} +  c_{-j} e^{-i(k_x x + k_y y)}, \ \ k_x,k_y \in \RR, \ \ k_x^2+ k_y^2 < 4,
\end{equation}
and a corresponding resonant pair $(\tau_{2j-1},\tau_{2j})$:
\begin{equation}\label{eq:res1}  
\begin{split}
  \tau_{2j-1}=p_{2j-1}+iq_{2j-1}, \ \ \tau_{2j} = p_{2j} + i q_{2j}, \ \ p_{2j}-p_{2j-1} = k_x, \ \ q_{2j}-q_{2j-1} = k_y, \\ |\tau_{2j-1}|=|\tau_{2j}|=1.
\end{split}  
\end{equation}
(Let us recall that for each unstable mode we have two resonant pairs associated with it:  $(\tau_{2j-1},\tau_{2j})$ and $(\tau_{2j+2N-1},\tau_{2j+2N})=(-\tau_{2j},-\tau_{2j})$.)
The restriction of ${\cal L}_0$ to ${\cal F}(p_{2j-1},q_{2j-1})={\cal F}(p_{2j},q_{2j})$ has a two-dimensional zero subspace, generated by the functions $f^{(+)}_{2j-1}$, $f^{(+)}_{2j}$,  where:
\begin{align}\label{eq:bloch3}
  f^{(\pm)}_{2j-1}&=\begin{bmatrix}1 \\ -i p_{2j-1} \pm q_{2j-1}  \end{bmatrix} \exp\left(i\left[ p_{2j-1} x + q_{2j-1} y\right]  \right), \\
  f^{(\pm)}_{2j}&=\begin{bmatrix}1 \\ -i p_{2j} \pm q_{2j}  \end{bmatrix} \exp\left(i\left[ p_{2j} x + q_{2j} y\right]  \right),\nonumber
\end{align}
Following \cite{Krichever2,Krichever3}, calculate the perturbation of the Riemann surface near this resonant pair. Denote by $\hat{\cal F}(\delta p,\delta q)$ the Bloch space with the multipliers:
\begin{equation}\label{eq:bloch4}
\tilde\varkappa_x = \exp{[iL_x (p_{2j-1}+\delta p) ]}, \ \ \tilde\varkappa_y = \exp{[iL_y (q_{2j-1} +\delta q)]}, \ \ |\delta p|\ll1, \ \  |\delta q|\ll1.
\end{equation}
In the leading order approximation in $\delta p$,  $\delta q$ one has:
\begin{equation}\label{eq:pert2}
\begin{split}  
  \tilde f_k^{(\pm)}=\begin{bmatrix} 1 \\ -i p_{2j-2+k} \pm q_{k} - i\delta p \mp  \frac{p_{k}}{q_{k}} \delta p  \end{bmatrix} \times \\ \times \exp\left(i\left[\left(p_{k}+\delta p\right)  x    + \left(q_{k}+ \delta q \right)  y\right]  \right), \ \ k=2j-1,2j.
 \end{split} 
\end{equation}
and
\begin{equation}\label{eq:pert4}
{\cal L}_0  \tilde f_{k}^{(+)} = \left(-\delta q - \frac{p_{k}}{q_{k}} \delta p \right) \tilde f_{k}^{(+)}.
\end{equation}
In the leading order approximations for the matrix representation of the block of ${\cal L}_1$ corresponding to this subspace we have:
\begin{equation}\label{eq:pert4.1}
  \begin{bmatrix} 0 & <f^{(+)}_{2j-1}|{\cal L}_1|f^{(+)}_{2j}> \\  <f^{(+)}_{2j}|{\cal L}_1|f^{(+)}_{2j-1}> & 0
  \end{bmatrix}   
\end{equation}
Let us calculate these matrix elements in  ${\cal L}_1$. The dual basic vectors $f_{k}^{(+)*}$, $k=2j-1,2j$  are defined by
\begin{align}\label{eq:pert5.0}
  f^{(+)*}_{k}(x+L_x,y) &= \exp{[-iL_x p_{k}]} f^{(+)*}_{k}(x,y), \\
  f^{(+)*}_{k}(x,y+L_y) &= \exp{[-iL_y q_{k}]} f^{(+)*}_{k}(x,y), \nonumber
\end{align}
and
\begin{equation}\label{eq:pert5.0}
 <f^{(+)*}_{k},f_{l}^{(+)}> = \delta_{k,l}, \ \  <f^{(+)*}_{k},f_{l}^{(-)}> =0, \ \ l=2j-1,2j;
\end{equation}
therefore
\begin{equation}\label{eq:pert5}
  f^{(+)*}_{k}=\frac{1}{2 L_x L_y q_{k}}\begin{bmatrix} i \bar\tau_{k}, 1  \end{bmatrix} \exp\left(-i\left[ p_{k} x + q_{k} y\right]  \right), 
\end{equation}
and
\begin{equation}\label{eq:pert6}
 <f^{(+)}_{2j-1}|{\cal L}_1|f^{(+)}_{2j}> =\frac{1}{2 q_{2j-1}}\begin{bmatrix} i \bar\tau_{2j-1}, 1  \end{bmatrix}
  \begin{bmatrix} 0 & c_{-j} \\ \bar c_j & 0  \end{bmatrix} \begin{bmatrix}1 \\ -i \tau_{2j}  \end{bmatrix} = -\alpha_j=\frac{\overline{c}_j+\bar\tau_{2j-1}\tau_{2j} c_{-j}}  {2q_{2j-1}},
\end{equation}
\begin{equation}\label{eq:pert7}
 <f^{(+)}_{2j}|{\cal L}_1|f^{(+)}_{2j-1}> =\frac{1}{2 q_{2j}}\begin{bmatrix} i \bar\tau_{2j}, 1  \end{bmatrix}
  \begin{bmatrix} 0 & c_{j} \\ \bar c_{-j} & 0  \end{bmatrix} \begin{bmatrix}1 \\ -i \tau_{2j-1}  \end{bmatrix}=\beta_j=\frac{\overline{c}_{-j}+\bar\tau_{2j}\tau_{2j-1} c_{j}}  {2q_{2j}}.
\end{equation}

Let  $(\tau_{2j-1},\tau_{2j})$ and $(\tau_{2j-1+2N},\tau_{2j+2N})=(-\tau_{2j-1},-\tau_{2j}) $ be two pairs of resonant points corresponding to the same monochromatic perturbation. Then we have the following symmetry:
\begin{equation}\label{eq:pert7.1}
\alpha_{j+N}\beta_{j+N} = \overline{\alpha_{j}\beta{j}}.
\end{equation}
Near the resonant pair $(\tau_{2j-1},\tau_{2j})$ the spectral curve in the leading order is defined by
\begin{equation}\label{eq:pert8}
\det  \begin{bmatrix}  -\frac{p_{2j-1}}{q_{2j-1}} \delta p -\delta q & -\varepsilon\alpha_j  \\ \varepsilon\beta_j    &  -\frac{p_{2j}}{q_{2j}} \delta p -\delta q \end{bmatrix} =0. 
\end{equation}
Using equation (\ref{eq:pert8}) we locally define $\delta q$ as a two-valued function of $\delta p$. Let us remark that $\delta p$ is a well-defined local parameter near the points $\tau_{2j-1}$ and $\tau_{2j}$, and it defines a local isomorphism between the neighborhoods of these points, and the map
\begin{equation}\label{eq:pert8.1}
(\delta p,\delta q)\rightarrow \delta p,
\end{equation}
is locally a two-sheeted covering. To calculate the branch point of this covering, denote:
\begin{equation}\label{eq:pert10}
\delta q= \delta\tilde q -\frac{1}{2}\left[ \frac{p_{2j-1}}{q_{2j-1}} + \frac{p_{2j}}{q_{2j}} \right]\delta p. 
\end{equation}
Equation (\ref{eq:pert8}) becomes:
\begin{equation}\label{eq:pert11}
\det  \begin{bmatrix}  - \frac{q_{2j}p_{2j-1}-q_{2j-1}p_{2j}}{2q_{2j-1}q_{2j}} \delta p -\delta\tilde q & -\varepsilon\alpha_j  \\ \varepsilon\beta_j    &  \frac{q_{2j}p_{2j-1}-q_{2j-1}p_{2j}}{2q_{2j-1}q_{2j}}   -\delta\tilde q \end{bmatrix} =0. 
\end{equation}
The branch points correspond to the double roots of (\ref{eq:pert11}) with respect to $\tilde\delta q$, therefore they correspond to the following values of $\delta p$:
\begin{equation}\label{eq:pert12}
\delta p= \pm \frac{2q_{2j-1}q_{2j}}{q_{2j}p_{2j-1}-q_{2j-1}p_{2j}}\varepsilon\sqrt{\alpha_j \beta_j}, \ \ \frac{p_{2j-1}}{q_{2j-1}} \delta p +\delta q = \pm\varepsilon\sqrt{\alpha_j \beta_j},
\end{equation}
Let us remark that
$$
q_{2j}p_{2j-1}-q_{2j-1}p_{2j}=\Im\left( \frac{\tau_{2j}}{\tau_{2j-1}} \right) =\Im\left(\tau_{2j} \bar\tau_{2j-1} \right).
$$

We assume that we fix one of the values of $\sqrt{\alpha_j \beta_j}$, and we use this value in all formulas below. For example, for generic data we may assume that $\Re\sqrt{\alpha_j \beta_j}>0$. Taking into account  (\ref{eq:unpert02}) we obtain the following formula for the branch points of this map in the leading order near the points $\tau_{2j-1}$ and $\tau_{2}$ respectively:

 \begin{align}\label{eq:pert12.2}
E_{4j-4+k}&= \tau_{2j-1}+(-1)^{k-1}\frac{2\tau_{2j-1}q_{2j}}{i\Im(\tau_{2j}\bar\tau_{2j-1})}\varepsilon\sqrt{\alpha_j \beta_j},\\
E_{4j-2+k}&= \tau_{2j-1}+(-1)^{k-1}\frac{2\tau_{2j}q_{2j-1}}{i\Im(\tau_{2j}\bar\tau_{2j-1})}\varepsilon\sqrt{\alpha_j \beta_j},\nonumber,
\end{align}
We assume here that at $E_{4j-3}$ and $E_{4j-1}$
\begin{align}\label{eq:pert12.2}
  \delta p = \left\{\begin{array}{ll} \frac{2q_{2j-1}q_{2j}}{\Im(\tau_{2j}\bar\tau_{2j-1})}\varepsilon\sqrt{\alpha_j \beta_j} & \mbox{ at } E_{4j-3} \sim  E_{4j-1}, \\
-\frac{2q_{2j-1}q_{2j}}{\Im(\tau_{2j}\bar\tau_{2j-1})}\varepsilon\sqrt{\alpha_j \beta_j} & \mbox{ at } E_{4j-2} \sim  E_{4j}.
                   \end{array}\right.   
\end{align}

The spectral curve for the perturbed operator is defined as follows. We cut the $\tau$-plane along the intervals $(E_{2j-1},E_{2j})$. For each resonant pair $(\tau_{2j-1},\tau_{2j})$ we glue the borders of the cuts $(E_{4j-3},E_{4j-2})$ and $(E_{4j-1},E_{4j})$. The point $E_{4j-3}$ is glued to $E_{4j-1}$, and the point  $E_{4j-2}$ is glued to $E_{4j}$. Moreover, if we glue a pair of points, the corresponding values of the Bloch multipliers $\varkappa_x$ are equal. The cycle $a_j$ is the oval surrounding the cuts $(E_{4j-1},E_{4j})$ and oriented counterclockwise, the cycle $c_j$ is the union of the oriented intervals $[E_{4j-3},0]$ and  $[0,E_{4j-1}]$, the cycles $b_j$ are defined by (\ref{eq:b-cycles}).
\begin{figure}[H]
  \centering{\includegraphics[width=0.8\textwidth]{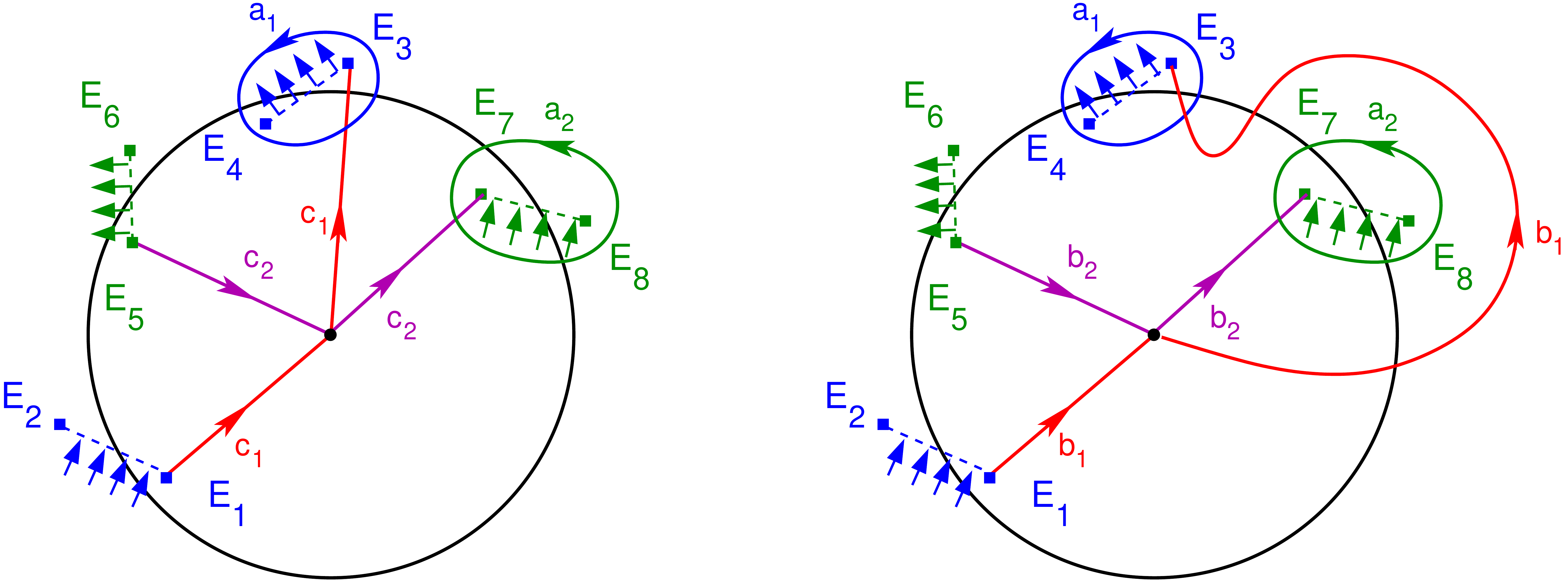}
      \caption{\small{\sl The perturbed curve. The borders of the cuts $(E_1,E_2)$ and $(E_5,E_6)$ are glued to the borders of the  cuts $(E_3,E_4)$ and $(E_7,E_8)$ respectively. On the left: the $a$-cycles and $c$-cycles. The system of $b$-cycles is shown on the right.}}\label{fig:perturbed10}}
\end{figure}

To calculate the basic differential $\omega_j$ in the leading order it is sufficient to know the positions of the points
$E_k$, $k=4j-3,\ldots,4j$, and the other branch points appear in higher-order corrections. On the corresponding elliptic curve we can use the following approximation:
\begin{equation}\label{eq:pert18}
\begin{split}  
  \omega_j = \frac{d\tau}{\sqrt{(\tau-E_{4j-1})(\tau-E_{4j})}} - \frac{d\tau}{\sqrt{(\tau-E_{4j-3})(\tau-E_{4j-2})}}=\\
  =d\log\left[\frac{\tau-\tau_{2j} + \sqrt{(\tau-E_{4j-1})(\tau-E_{4j})}}{\tau-\tau_{2j-1}+ \sqrt{(\tau-E_{4j-3})(\tau-E_{4j-2})}}\right].
 \end{split} 
\end{equation}
In (\ref{eq:pert18}) we assume that if $\tau-\tau_{2j-1}$ is of order 1, then $\sqrt{(\tau-E_{4j-3})(\tau-E_{4j-2})}\sim 
\tau-\tau_{2j-1}$. Analogously, if $\tau-\tau_{2j}$ is of order 1, then then $\sqrt{(\tau-E_{4j-2})(\tau-E_{4j-1})}\sim 
\tau-\tau_{2j}$, therefore outside the neigbourhood of this resonant pair, formula (\ref{eq:pert18}) coincides with
(\ref{eq:unpert01}).
It is clear that in the leading order
\begin{equation}\label{eq:pert18.1}  
  \omega_j =\left\{\begin{array}{ll}d\log\left[\tau-\tau_{2j} + \sqrt{(\tau-E_{4j-1})(\tau-E_{4j})}\right], & \tau\sim\tau_{2j},\\
                 -d\log\left[\tau-\tau_{2j-1} + \sqrt{(\tau-E_{4j-2})(\tau-E_{4j-2})}\right], & \tau\sim\tau_{2j-1};
                   \end{array}
                   \right.
\end{equation}
therefore, in the leading order approximation, the basic holomorphic differential $\omega_j$ at the handle connecting $\tau_{2j-1}$ with $\tau_{2j}$ can be written as
\begin{equation}\label{eq:pert9}
\omega_j = d\log \left(\frac{p_{2j-1}}{q_{2j-1}} \delta p + \delta q\right) = - d\log \left(\frac{p_{2j}}{q_{2j}} \delta p + \delta q\right).
\end{equation}

The divisor points are defined by the condition: the first component of the Bloch eigenfunction for the perturbed operator is equal to zero at $z=0$, or, equivalently
\begin{equation}\label{eq:pert13}
\begin{bmatrix}  - \frac{p_{2j-1}}{q_{2j-1}} \delta p -\delta q & -\varepsilon\alpha_j  \\ \varepsilon\beta_j    &  -\frac{p_{2j}}{q_{2j}}   -\delta q \end{bmatrix} \begin{bmatrix} 1 \\ -1 \end{bmatrix} = \begin{bmatrix} 0 \\ 0 \end{bmatrix};
\end{equation}
therefore for the divisor point $\gamma_j$ we have
\begin{equation}\label{eq:pert14}
\frac{p_{2j-1}}{q_{2j-1}} \delta p +\delta q = \varepsilon\alpha_j.
\end{equation}
Combining (\ref{eq:pert12}) with (\ref{eq:pert14}) we obtain: 
\begin{lemma}
Let $E$ denote one of the branch points obtained by perturbing the resonant pair $(\tau_{2j-1},\tau_{2j})$, and let us choose it as the starting point of the Abel transform.  Then for the Abel transform of the divisor point $\gamma_j$ is given in the leading order by the following formula: 
\begin{equation}\label{eq:pert15}
\left[\vec A_{E_{4j-3}}(\gamma_j)\right]_k = \left\{ \begin{array}{ll} 0 & k\ne j \\ \log\left[\frac{\alpha_j}{\sqrt{\alpha_j\beta_j}} \right] & k=j. \end{array} \right.
\end{equation}
\end{lemma}
\begin{remark}  
It is easy to check that, for $j\le N$, the point $\sigma\gamma_{j+N}$ can be defined by:
\begin{equation}\label{eq:pert16}
\begin{bmatrix}  - \frac{p_{2j-1}}{q_{2j-1}} \delta p -\delta q & -\varepsilon\alpha_j  \\ \varepsilon\beta_j    &  -\frac{p_{2j}}{q_{2j}}   -\delta q \end{bmatrix} \begin{bmatrix} \tau_{2j} \\ -\tau_{2j-1} \end{bmatrix} = \begin{bmatrix} 0 \\ 0 \end{bmatrix};
\end{equation}
therefore, in the leading order,  
 \begin{equation}\label{eq:pert17}
\left[\vec A_{E_{4j-1}}(\sigma\gamma_{j+N})\right]_k = \left\{ \begin{array}{ll} 0 & k\ne j, \\ \log\left[\frac{\tau_{2j-1}\alpha_j}{\tau_{2j}\sqrt{\alpha_j\beta_j}} \right] & k=j. \end{array} \right.
\end{equation}
Therefore, from (\ref{eq:pert15}),  (\ref{eq:pert17}), (\ref{eq:abel03}) it follows that, in the leading order, the reality condition (\ref{eq:cher0}) is fulfilled.
\end{remark}

In this approximations for the Abel transform with the starting point $\tau=0$ we obtain:
\begin{align}\label{eq:pert19}
  A_j(E_{4j-3}) &= -\log\left[\frac{\tau_{2j} q_{2j}}{i\Im(\tau_{2j}\tau_{2j-1}^{-1})(\tau_{2j-1}-\tau_{2j}) }\varepsilon\sqrt{\alpha_j\beta_j}\right], \\
   A_j(E_{4j-1}) &= \log\left[\frac{\tau_{2j-1} q_{2j-1}}{i\Im(\tau_{2j}\tau_{2j-1}^{-1})(\tau_{2j}-\tau_{2j-1}) }\varepsilon\sqrt{\alpha_j\beta_j}\right], \
\end{align}
\begin{equation}\label{eq:pert20}
b_{jj}=A_j(E_{4j-1})- A_j(E_{4j-3})  =\log\left[\frac{\tau_{2j-1}\tau_{2j} q_{2j-1}q_{2j}}{\Im^2(\tau_{2j}\tau_{2j-1}^{-1})(\tau_{2j-1}-\tau_{2j})^2}\varepsilon^2(\alpha_j\beta_j)\right].
\end{equation}

{\color{Violet}
\begin{remark}
The asymptotics of the Abel differentials and the Riemann matrix for Riemann surfaces with pinched cycles is discussed in \cite{Fay}; moreover, using the technique of this book, it is possible to calculate the second-order corrections. For surfaces close to the degenerate ones, one can effectively use the Schottky parameterization, see book \cite{BBEIM} and references therein for more details.
\end{remark}
\begin{remark}
If one of the quantities $\alpha_j$, $\beta_j$  (\ref{eq:pert12.2}) equals 0, the corresponding resonant point becomes double point in the leading order approximation. Using the arguments analogous to \cite{Krichever2,Krichever3}, it is easy to check that a regular doubly-periodic perturbation may be chosen in such a way that we obtain a double point in the exact theory. For the 2-d Schr\"odinger operator at a fixed energy level the existence of singular spectral curves corresponding to regular doubly-periodic potentials was pointed out in  \cite{Krichever2}. Let us point out that non-removable double points correspond to resonant pairs associated to unstable modes; the resonant points associated to stable modes after perturbation become either small handles or remain removable double points. The role of the singular spectral curves in the theory of soliton equations including the DS2 and the modified Novikov-Veselov equation was studied in \cite{Taim4,Taim7}.
\end{remark}
}
\section{Vector of Riemann constants}

Denote the starting point of the Abel transform by $P_0$. Then the vector of Riemann constant is given by the following formuls (see \cite{Dubrovin2}):
\begin{equation}
\label{eq:riemann_const1}
K_j = \frac{b_{jj}}{2} - \pi i - \frac{1}{2\pi i}\sum_{k\ne j }^{2N} \int_{\tilde a^k}  A_j(\gamma) \omega_k,
\end{equation}
where $A_j$ is the $j$-th component of the Abel transform of the point $\gamma$. 
\begin{equation}
\label{eq:riemann_const2}
{\cal A}_j(\gamma) =\int_{P_0}^{\gamma} \omega_j.
\end{equation}
In this formula it is important to have a proper realization of the basis of cycles. More precisely, it is necessary that:
\begin{enumerate}
\item All basic cycles start and end at the same starting point $Q_0$;
\item They do not intersect outside the point $Q_0$.  
\item Near the point $Q_0$ we have the following order of curves (if counted clockwise): Start of $a_1$, end of $b_1$, end of $a_1$,
  start of $b_1$, start of $a_2$, end of $b_2$, end of $a_2$, start of $b_2$, \ldots, start of $b_g$. See Fig.~\ref{fig:riemann1} left for $g=8$:
\end{enumerate}
\begin{figure}[H]
  \centering{
    \includegraphics[width=0.2\textwidth]{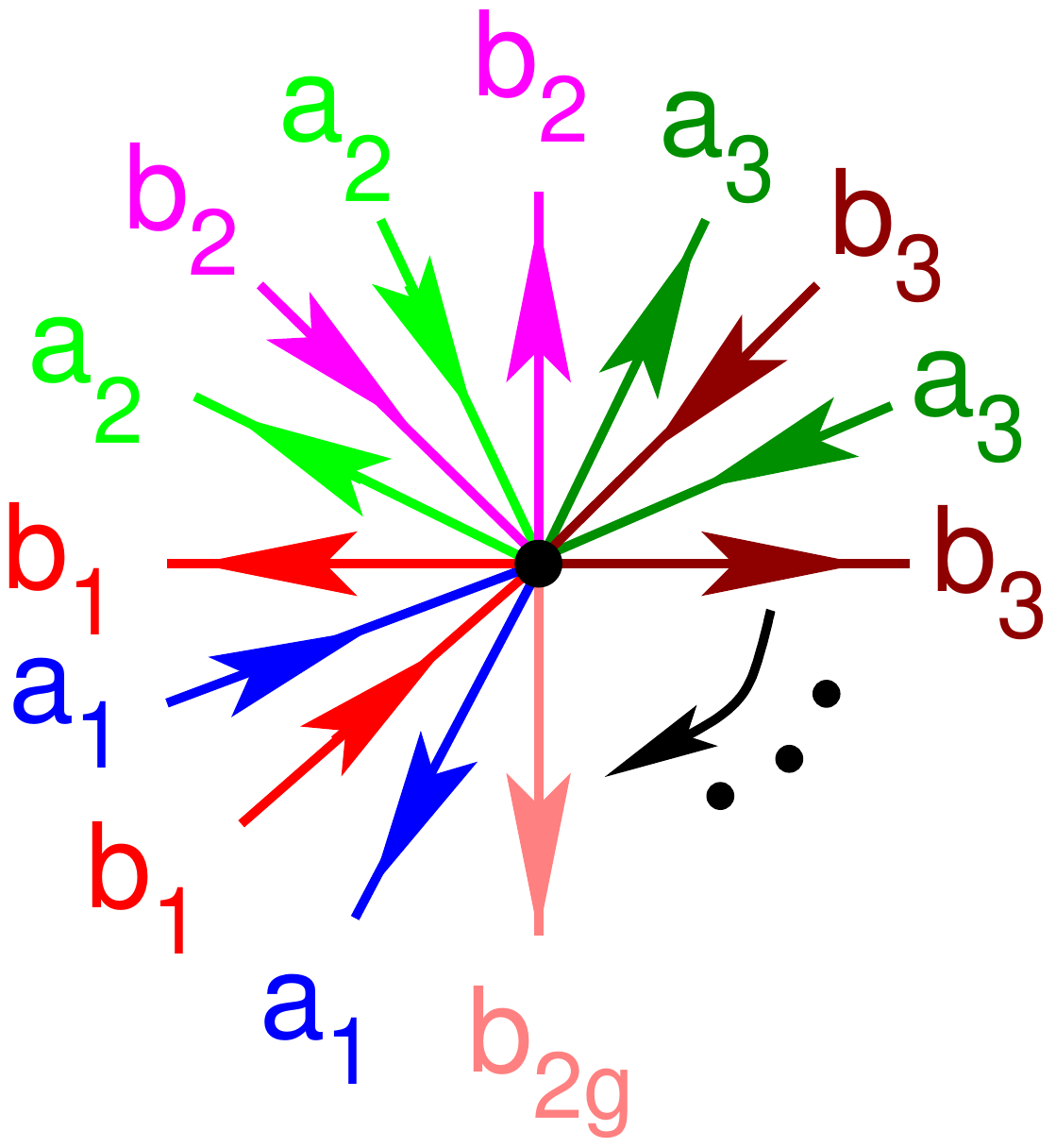}\hspace{5mm}
    \includegraphics[width=0.35\textwidth]{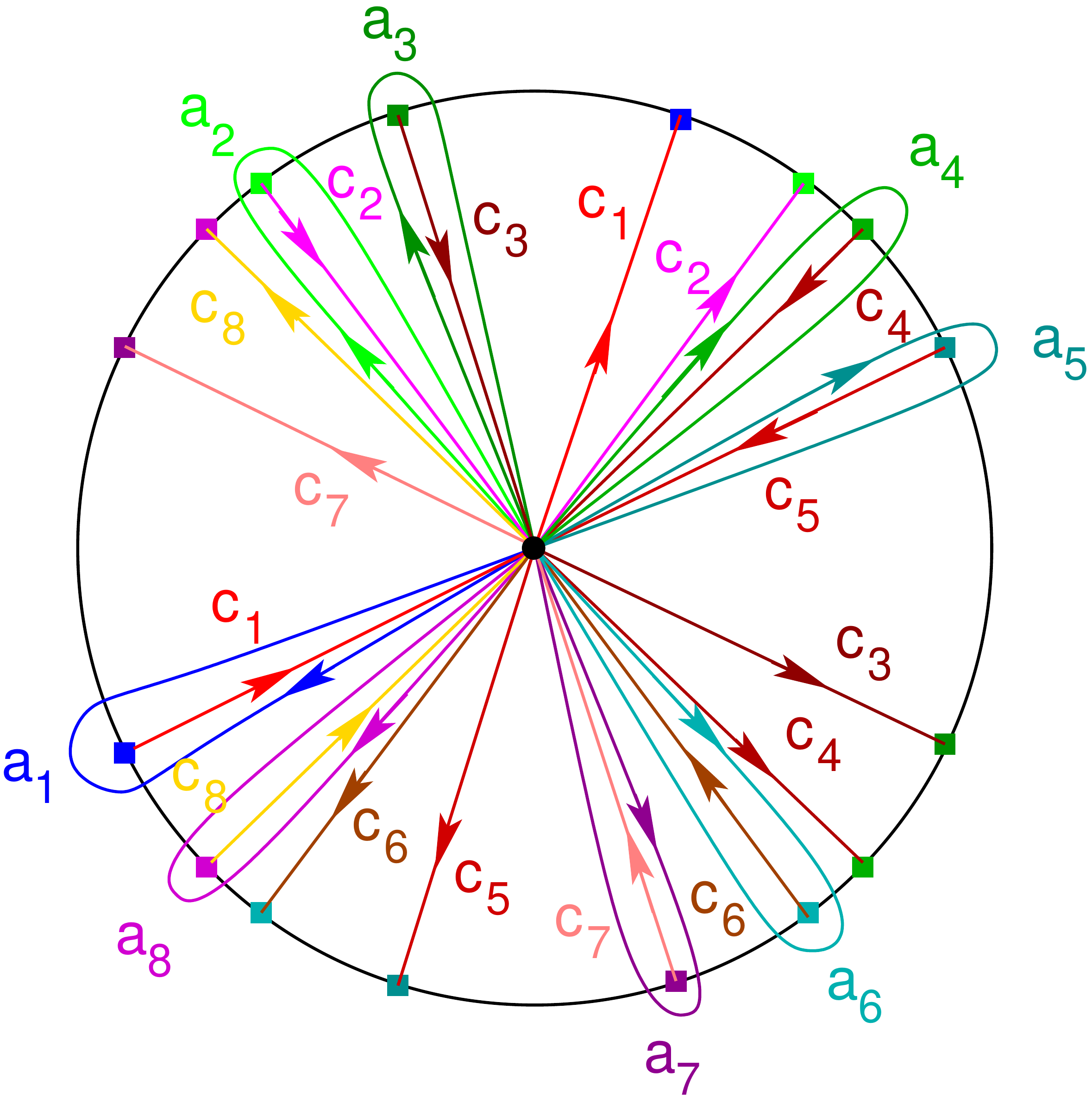}\hspace{5mm}
    \includegraphics[width=0.35\textwidth]{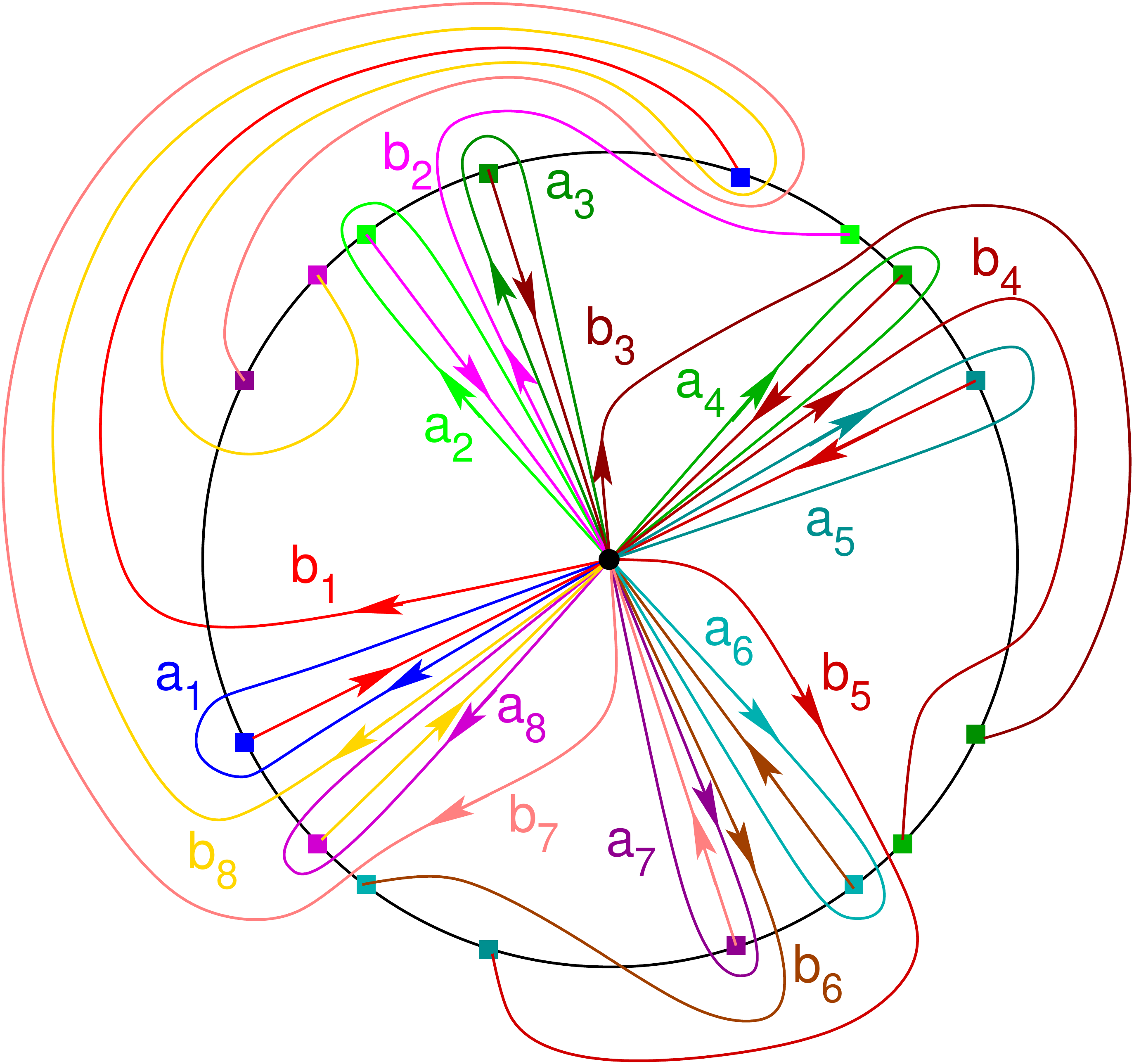}
      \caption{\small{\sl On the left: the order of cycles near the point $Q_0$. On the right: the collection of contours used in the calculation of the Riemman constants.}}\label{fig:riemann1}}
\end{figure}

A basis of such cycles is presented at Figure~\ref{fig:riemann1}. It is clear that, for $j \ne k$, in the $\epsilon$-neighbourhood of the cut $[E_{4k-3},E_{4k-2}]$ we have 
\begin{equation}
\label{eq:riemann_const3}
A_j(\gamma) = A_j(E_{2k-1})+ O(\epsilon);
\end{equation}
therefore 
\begin{equation}
\label{eq:riemann_const4}
- \frac{1}{2\pi i}\int_{\tilde a^k}  A_j(\gamma) \omega_k = - A_j(E_{2k-1})+ O(\epsilon).
\end{equation}
Therefore by analogy with \cite{GS5} we can use the following modification of the formulas: Let the divisor point $\gamma_k$ be located near the contour $a_k$. Then we redefine the Abel transform of the divisor by assuming
\begin{equation}
\label{eq:abel11}
A_j(\gamma_k) =\int_{E_{4k-3}}^{\gamma_k}\omega_j.
\end{equation}
Then
\begin{equation}
\label{eq:riemann_const1}
K_j = \frac{b_{jj}}{2} - \pi i + A_j(E_{4j-3}) +O(\epsilon).
\end{equation}

{\color{Violet}
\section{Summary of the results of the paper}
  
  The results of the paper can be summarized as follows:

Consider the focusing Davey-Stewardson 2 equation (\ref{eq:DS2}) in the space of spatially doubly-periodic functions with the periods $L_x,L_y$ respectively. The Cauchy problem for (\ref{eq:DS2}) is called \textbf{doubly-periodic Cauchy problem for anomalous waves} if the Cauchy data are a small doubly-periodic perturbation of the constant background (\ref{eq:DS2_1}). 
\begin{theorem}
  Assume that the Cauchy data  (\ref{eq:DS2_1}) for equation (\ref{eq:DS2}) have the following properties:
 \begin{enumerate}
 \item The background is unstable, i.e. the open disk $k_x^2+k_y^2 < 4 |a|^2$ contains a least one point $(k_x,k_y)$ of the type (\ref{eq:one_harm}) with $(n_x,n_y)\ne(0,0)$;
 \item The periods $L_x,L_y$ are generic in the following sense:
   \begin{itemize}
   \item $k_x^2+k_y^2 \ne 4 |a|^2$ for all integer $n_x,n_y$;
   \item Consider the map $\CC\rightarrow \CC^2$, $\tau\rightarrow (\varkappa_x,\varkappa_y)$ from the $\tau$-plane to the quasimomenta space  defined by (\ref{eq:param}), (\ref{eq:DS_lax4.4}). Then for each point $(\varkappa_x,\varkappa_y)\in\CC^2$ the number of preimages is not greater than 2. Equivalently, this condition means that all multiple points of the spectral curve in the quasimomenta space are double points;
   \end{itemize}  
 \item The Cauchy data satisfy the following genericity conditions: for each unstable mode we have $\alpha_j\beta_j\ne 0$, where the quantities $\alpha_j$, $\beta_j$ are expressed through the corresponding Fourier coefficient of the perturbation by (\ref{eq:pert6}),  (\ref{eq:pert7}), respectively.
 \end{enumerate}

To simplify the final formulas, assume also that $a=1$ and the zero Fourier harmonics of the perturbation $v_0(x,y)$ is equal to zero (\ref{eq:zero_harm}). Due to the scaling invariance of the DS equation these constraints are not restrictive. 

Then, for $|\varepsilon|\ll 1$, the leading order solution of the Cauchy problem  (\ref{eq:DS2_1}) for  the focusing DS2 equation is provided by formula (\ref{eq:psi3}), where ${\cal C}_z= {\cal C}_{\bar z} = {\cal C}_t =0$, the Riemann $\theta$-function is defined by (\ref{eq:theta0}), $g=2N$, where $N$ is the number of unstable modes (while counting the unstable modes we assume $(n_x,n_y)\ne(0,0)$), the Riemann matrix of periods $B=(b_{jl})$ is defined by (\ref{eq:pert20}), (\ref{eq:abel04}), the vectors  $\vec W_z$, $\vec W_{\bar z}$, $\vec W_{t}$ are given by (\ref{eq:abel01}), (\ref{eq:abel02}), $\vec A(\infty_1)=\vec 0$,  $\vec A(\infty_2)$ is given by  (\ref{eq:abel03}), $\vec A({\cal D})$ is given by (\ref{eq:pert15}), $\vec K$ is given by  (\ref{eq:riemann_const1}), (\ref{eq:pert20}).

Moreover, following the scheme of \cite{GS5}, it is sufficient to keep only the leading order terms in (\ref{eq:theta0}), and final formulas can be expressed in terms of elementary functions of the Cauchy data. The leading order terms are different for different regions in the $(x,y,t)$ space; therefore the approximating formulas depend on the region. Following  \cite{GS5}, the boundaries of these regions can be calculated explicitly in terms of the Cauchy data.
\end{theorem}

\section{Acknowledgments} The work P. G. Grinevich was supported by the Russian Science Foundation under grant No. 21-11-00311. P. M. Santini acknowledges support from the Italian Ministry for Research through the PRIN2020 project (number 2020X4T57A). \\
We acknowledge useful discussions with A. Bogatyrev.

}

\end{document}